\documentclass[svgnames]{ceurart}
\usepackage{llncs-dropin}
\usepackage{preamble}
\setboolean{arxiv}{true}
\pratendSetGlobal{normal}
\makeatletter
\def\@oddfoot{\hfil\thepage\hfil}
\def\@evenfoot{\hfil\thepage\hfil}
\makeatother
\newcommand{\tikzscale}[1]{\scalebox{.8}{#1}}
\begin{document}
\copyrightyear{2026}
\copyrightclause{Copyright for this paper by its authors.
  Use permitted under Creative Commons License Attribution 4.0
  International (CC BY 4.0)}

\ifarxiv
\conference{Full version of \cite{GCM-original}, presented at Graph Computational Models 2026}
\title{Structural Morphisms for Nested Conditions --- Full Version}
\else
\conference{Joint Proceedings of the STAF 2026 Workshops: AgileMDE, GCM, ICMM, LLM4SE, TTC. Rennes, France, June 29-July 3, 2026}
\title{Structural Morphism for Nested Conditions}
\fi

\author[1]{Arend Rensink}[email=arend.rensink@utwente.nl]
\author[2]{Andrea Corradini}[email=andrea.corradini@unipi.it]

\address[1]{University of Twente, Netherlands}
\address[2]{University of Pisa, Italy}

\begin{abstract}
Nested conditions are used, among other things, as a graphical way to express first order formulas ruling the applicability of a graph transformation rule to a given match. In this paper, we first introduce several operators on conditions mimicking logical connectives. Next we propose an original notion of structural morphism among nested conditions, and we identify circumstances under which morphisms are consistent with the entailment of the corresponding conditions. Finally we frame the results in a categorical context, proving functoriality and universality properties of the various operations. 
\ifarxiv

\medskip\noindent This is the full version (including proofs) of the paper published as \cite{GCM-original}.
\fi
\end{abstract}


\maketitle
\section{Introduction}
\label{sec:introduction}

Representing formulas of First-Order Logic (FOL) by graphs or more general graphical structures was explored in various areas of Theoretical Computer Science and Logics along the decades. A canonical example is represented by edge-labelled graphs, which can be regarded as an alternative syntax for formulas of a fragment of FOL including just basic (binary) predicates, equality, conjunction and existential quantification: we call this briefly the $\exists$-fragment.

Let $\Lab = \{\gl{a}, \gl{b}, \gl{c}, \ldots\}$ be a set of binary relation symbols, which we shall use also as edge labels. As an example, let $A$ be the graph $\inline{\twoedge{x}{a}{y}{b}{z}}$. We consider this graph as a sound representation of the formula  $\phi_A = \exists x,y,z\st \gl{a}(x,y) \land \gl{b}(y,z)$, in the following sense: a graph $G$ satisfies $\phi_A$ (or is a \emph{model} of $\phi_A$) if and only if there is a graph morphism $h$ from $A$ to $G$. 
Now consider $B = \inline{\oneedge{x}{a}{y}}$, representing formula $\phi_{B} = \exists x,y\st \gl{a}(x,y)$. Graph $B$ has an obvious inclusion morphism into $A$, viz.\ $i\of B \to A$. Therefore every morphism $h\of A \to G$ gives rise to a composed morphism $i;h\of B\to G$,\footnote{Along the paper, we denote by $f;g\of A \to C$ the composition of arrows $f\of A \to B$ and $g \of B \to C$, using diagrammatic order.} implying that every graph that satisfies $\phi_A$ also satisfies $\phi_{B}$, thus $\phi_A$ \emph{entails} $\phi_{B}$ (written $\phi_A \entails \phi_{B}$). It is worth noting that, in this elementary framework, graph morphisms can represent both a satisfaction relation (between formulas and models) and an entailment relation among formulas.

These concepts were exploited for example by Chandra and Merlin in \cite{DBLP:conf/stoc/ChandraM77} in the framework of relational database queries. They show there that every \emph{conjunctive query} (a formula of the $\exists$-fragment, but with relations of any arity) has a natural model, viz.\ a graph, and query inclusion is equivalent to the existence of a graph homomorphism between those natural models. Therefore morphisms are not only sound, but also complete with respect to entailment, and it follows that query inclusion is decidable, even if NP-complete: an interesting logical result obtained with graph theoretical techniques.

In the realm of Graph Transformation Systems  (GTSs)~\cite{eept:fundamentals-agt,DBLP:books/sp/HeckelT20},
 the need of representing formulas by graphs arose in a natural way. Indeed,
in any approach a rule consists of at least two graphs, $L \leadsto R$, and to apply it to a graph $G$, first a morphism $m: L \to G$ has to be found. 
 By the above discussion, we can consider $L$ as a formula of the $\exists$-fragment that has to be satisfied by $G$, as an application condition of the rule. 

It soon turned out that in order to use GTSs for even simple specifications, more expressive application conditions were needed. 
In~\cite{NegativeAC} the authors introduced \emph{Negative Application Conditions (NACs)}, allowing to express (to some extent) negation and disjunction.  A NAC $\cN$ is  a finite set of morphisms from $L$, $\cN = \{n_i: L \to Q_i\}_{i\in[1,k]}$, and a morphism $m: L \to G$ \emph{satisfies} $\cN$ if for all $i\in[1,k]$ there is no morphism $m_i: Q_i \to G$ such that $n_i;m_i = m$. 
It follows that such a NAC represents a formula of the shape $\exists \bar{x}\st \phi_L \wedge \neg (\exists \bar{y}_1\st \phi_{Q_1} \vee \ldots \vee \exists \bar{y}_n\st \phi_{Q_n})$: as in \cite{Rensink-FOL}, we call this the $\exists \neg \exists$-fragment (of FOL).

Note that differently from the $\exists$-fragment, the structures representing the $\exists \neg \exists$-fragment are no longer graphs, but diagrams (``stars'') in \cat{Graph}, the category of graphs; and satisfaction does not require just the existence of a matching morphism from $L$, but also the non-existence of certain other morphisms.
NACs were generalized in~\cite{Rensink-FOL,Habel-FOL} to \emph{Nested (Application) Conditions}, where the structure of a condition is a finite tree of arbitrary depth rooted at $L$, and satisfaction is defined like for NACs, but iterating further at each level of the tree. Interestingly, nested conditions were proved to have the same expressive power of full FOL.

Since nested conditions denote formulas, they are the objects of an obvious category (actually, a preorder) where arrows represent entailment: paper~\cite{bchk:conditional-reactive-systems} includes a comprehensive presentation of this preorder and of the categorical properties of logical operators on conditions.\footnote{Nested conditions are defined in a quite different way in this paper, but the expressive power is the same.} But differently from the case of the $\exists$-fragment, where entailment can be ``explained" by the existence of a graph morphism, we are not aware of similar results for the larger $\exists \neg\exists$-fragment nor for FOL. More explicitly, despite the fact that NACs first and Nested Conditions next were defined as suitable diagrams in a category of graphs (or of similar structures), we are not aware of  definitions of \emph{structural morphisms} among such application conditions, providing evidence for (some cases of) entailment like simple graph morphisms do for the $\exists$-fragment. We addressed this issue for the first time in~\cite{DBLP:conf/birthday/Rensink025} by considering two classes of nested conditions, the \emph{arrow-based} ones (those of~\cite{Rensink-FOL}) and the more refined \emph{span-based} ones, introducing various kinds of morphisms for them. In this paper we stick to the first kind of conditions, but we propose  a more systematic and general approach, which is the main contribution of this paper. In particular, the morphisms we define here at the same time are simpler than and subsume those of \cite{DBLP:conf/birthday/Rensink025} for arrow-based conditions. Their generalization to span-based conditions is ongoing work.  

We start by recalling in \cref{sec:conditions} the main definitions related to nested conditions and their satisfaction from \cite{Rensink-FOL}, and we present several constructions on them aimed at capturing logical operators in a categorical setting. The section is concluded by a completeness result showing that every condition can be built in a canonical way using the proposed operators. 
\Cref{sec:morphisms} presents the main original contribution, i.e.~a very general and weak notion of structural morphisms among conditions, which satisfies the desired properties of being composable and having identities. Such morphisms include collection of arrows between the source and target condition, flipping direction at each nesting level, and forming certain (possibly not commuting) squares with the arrows of the conditions themselves. By imposing suitable conditions on such squares, we identify two subclasses of morphisms, \emph{reflective} and the \emph{preservative} ones, which are shown to reflect and preserve satisfaction, respectively, and thus witness entailment between the formulas associated with the related conditions. We also show that the composition of morphisms of the same kind is again a morphism of that kind, and that identities are morphisms of both kinds. Finally we show that under mild assumptions for each preservative morphism there is a reflective one in the opposite direction.
In \cref{sec:categories} we first frame the definitions of the various kind of morphisms of \cref{sec:morphisms} and their relations in a categorical context, introducing several categories of conditions and identity-on-object functors among them. Next, this allows us to state and prove some functoriality and universality properties of the constructions on conditions of \cref{sec:conditions}. \Cref{sec:conclusion} summarizes the contributions of the paper, discusses some related work and hints at future developments.
 
\section{Nested conditions and satisfaction}
\label{sec:conditions}

We start by recalling some basic concepts and the standard notion of nested condition from~\cite{Rensink-FOL}. We will drop the adjective ``nested'' and simply speak of \emph{conditions}, as we will not consider any other kind of conditions along the paper. 

\noindent
\textbf{Assumption}\ 
For the sake of generality, along the paper formal definitions and results will be phrased in terms of objects and arrows of a generic category $\bC$ that we assume to be a \emph{category of presheaves}, i.e., a category of contravariant functors from a small category \cat{S} to \cat{Set}, thus \cat{C} $= [\op{\cat{S}} \to \cat{Set}]$. Categories of presheaves are known to be toposes~\cite{MacLaneMoerdijk1992}: we will call them \emph{presheaf toposes}.
\medskip

Assuming that \cat{C} is a presheaf topos ensures several properties we need in the constructions of this paper: in particular, that all limits and colimits exist (and can be computed pointwise), and also that epis are stable under pullback~\cite{MacLaneMoerdijk1992}. Furthermore, \cat{C} is   \emph{adhesive}~\cite{ls:adhesive-journal}, enjoying several properties exploited in the algebraic theory of graph rewriting, where the results of this paper have potential interesting applications. Note that requiring \cat{C} to be just adhesive would not suffice: for example, we need arbitrary pushouts, while adhesivity only guarantees pushouts along monos.

\begin{example}[presheaf toposes]\label{ex:presheaf-toposes}

%
  Several categories of graphs and hypergraphs are presheaf toposes, including \emph{graph structures}~\cite{EhrigHKLRWC97}. For example, directed unlabelled graphs are obtained with \cat{S} the free category generated by $\mygraph{
  \node (1) {$\bullet$};
  \node (2) [right=of 1] {$\bullet$};
  \path (1) edge[bend left=20,->] (2)
        (1) edge[bend right=20,->] (2);
}$. 
Furthermore, presheaf toposes are closed under the construction of slice and functor categories (see Sec.~5 of~\cite{AzziCR19} for a recap), thus they also include \emph{triple graphs}~\cite{Anjorin_Leblebici_Schürr_2016}, \emph{typed graphs}~\cite{eept:fundamentals-agt}, and in particular $\cat{Graph}$, the category of directed, edge-labelled multigraphs introduced below, that will be used for examples and intuitions in the rest of the paper.   
\end{example}


\begin{definition}[Category of Graphs]\label{def:graph-and-morphism}
  A \emph{(directed, edge-labelled) graph}~$G = (V_G,E_G,s_G,t_G,\ell_G)$ has sets~$V_G$ of nodes and~$E_G$ of edges, along with source and target functions~$s_G,t_G\of E_G\to V_G$, and an edge-labelling function $\ell_G\of E_G \to \Lab$. 
  A \emph{graph morphism}~$f:G\to H$ is a pair of functions~$f = (f_V\of V_G\to V_H, f_E\of E_G\to E_H)$ that preserve incidence (that is, $s_G; f_V = f_E;s_H$ and $t_G; f_V = f_E ; t_H$) and labels ($f_E;\ell_H = \ell_G$).
  Graphs and graph morphisms determine  category~$\cat{Graph}$.\qed
  
\end{definition}

\medskip\noindent
Conditions are inductively defined as follows:

\begin{definition}[condition]\label{def:condition}
  Let $R$ be an object of $\bC$. $\AC R$ (the set of \emph{conditions} over $R$) and $\AB R$ (the set of \emph{branches} over $R$) are the smallest sets such that
  \begin{itemize}[nosep]
  \item $c\in \AC R$ if $c=(R,p_1\ccdots p_w)$ is a pair with $p_i\in \AB R$ for all $1\leq i\leq w$, where $w \geq 0$;
  \item $p\in \AB R$ if $p=(a,b)$ where $a: R\to P$ is the \emph{branch arrow} and $b\in \AC P$ the \emph{branch condition}.\qed
  \end{itemize}
\end{definition}
\noindent
We call $R$ the \emph{root} of a condition or branch, and $P$ the \emph{pattern} of a branch (which is simultaneously the root of its branch condition). \Cref{fig:condition} provides a visualisation of a condition $c$. We use $c$ to range over conditions and $p,q$ to range over branches. We use $|c|=w$ to denote the width of a condition $c$, $R^c$ to denote its root, and $p^c_i=(a^c_i,b^c_i)$ its $i$-th branch; for a condition $c_j$, the notations are $R_j$ for the root and $p^j_i=(a^j_i,b^j_i)$ for the $i$-th branch. Finally, we use $P^c_i$ ($=R^{b_i}$) for the pattern of branch $p^c_i$.
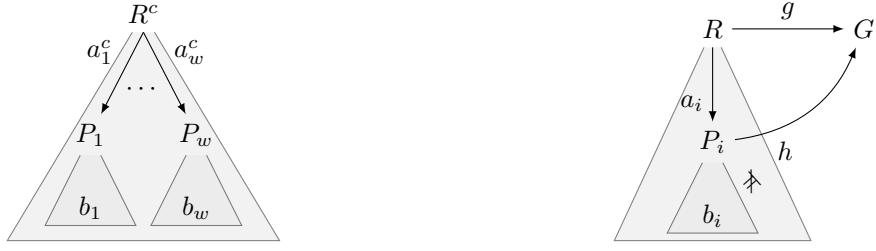
\begin{figure}[t]
\centering
\subcaptionbox
  {Condition $c=(R,p_1\ccdots p_w)$, with $p_i=(a_i,b_i)$ for $1\leq i\leq w$
   \label{fig:condition}}
  [.45\textwidth]
  {\begin{tikzpicture}[on grid]
  \node (Rc) {$R^c$};
  \tri{Rc}{3.0}{1.8}{3.6}{}
  \node (P1) [below left=1.6 and .7 of Rc] {$P_1$};
  \node [below=of Rc] {$\cdots$};
  \node (Pn) [below right=1.6 and .7 of Rc] {$P_w$};

  \tri{P1}{1.2}{.6}{1.2}{$b_1$}
  \tri{Pn}{1.2}{.6}{1.2}{$b_w$}
  
  \path (P1) edge[<-] node[above left] {$a^c_1$} (Rc.south)
        (Pn) edge[<-] node[above right] {$a^c_w$} (Rc.south);
\end{tikzpicture}}
\quad
\subcaptionbox
  {$g\sat c$, with responsible branch $p_i=(a_i,b_i)$ and witness $h$ such that $g=a_i;h$
   \label{fig:satisfaction}}
  [.5\textwidth]
  {\begin{tikzpicture}[on grid]
  \node (R) {$R$};
  \tri{R}{2.8}{1.3}{2.6}{}
  \node (Pi) [below=1.5 of R] {$P_i$};
  \tri{Pi}{1.2}{.6}{1.2}{$b_i$}
  \node (G) [right=2 of R] {$G$};

  \path (R) edge[->] node[above] {$g$} (G)
        (R) edge[->] node[left,near end] {$a_i$} (Pi)
        (Pi) edge[->,bend right] node[pos=0.2,below right] (h) {$h$} (G)
        (h) edge[draw=none] node[sloped,allow upside down] {$\nsat$} (Pi-label);
\end{tikzpicture}}
\caption{Visualisations for  conditions}
\end{figure}

Note that, as a consequence of the inductive nature of \cref{def:condition}, every condition has a finite \emph{depth} $\depth(c)$, defined as $0$ if $|c|=0$ and $1+\max_{1\leq i\leq |c|} \depth(c_i)$ otherwise. The depth will provide a basis for inductive proofs.

\begin{example}\label{ex:conditions}
\fcite{conditions} shows the graphical representation of three conditions, rooted in the discrete one-node graph \inline{\onenode x}. The graph morphisms are in all cases implied by the graph structure and node names. The leftmost one, $c_1 \in \AC{\inline{\onenode x}}$, is defined according to Def.~\ref{def:condition} as 

\smallskip
$\begin{array}{ll}
c_1 = (\inline{\onenode x}, (f,c_{11})) & c_{11} = (\inline{\oneloopleft{x}{b}}, (g,c_{111})(h,c_{112}))\\
c_{111} = (\inline{\looponeedge{x}{b}{a}{y}}, \epsilon) & c_{112} = (\inline{\looponeedge{x}{b}{c}{y}}, \epsilon)
\end{array}$
\smallskip

\noindent
For the depths,  $\depth(c_{111}) = \depth(c_{112}) = 0$, $\depth(c_{11}) = 1$ and $\depth(c_1) = 2$.

Based on the notion of satisfaction that we are about to introduce, every condition can be seen as a property on morphisms from its root to a graph, expressed in the first-order logic on graphs, see~\cite{Rensink-FOL}.  
Due to space constraints, we do not discuss the precise correspondence between conditions and FOL formulas here, but we present the formulas that correspond to the conditions of \cref{fig:conditions}:
\begin{itemize}[nosep]
\item $c_1$ is equivalent to $\lb(x,x)\wedge \neg \exists y\st(\la(x,y)\vee \lc(x,y))$
\item $c_2$ is equivalent to $\exists y\st \lb(x,y) \wedge \neg \la(y,y)\wedge \neg \exists z\st \lc(y,z)$ 
\item $c_3$ is equivalent to $\la(x,x)\vee (\exists y\st \lb(x,y) \wedge (\forall v,z\st \lc(y,v)\wedge \lc(y,z) \rightarrow v=z))$\qed
\end{itemize}
\end{example}
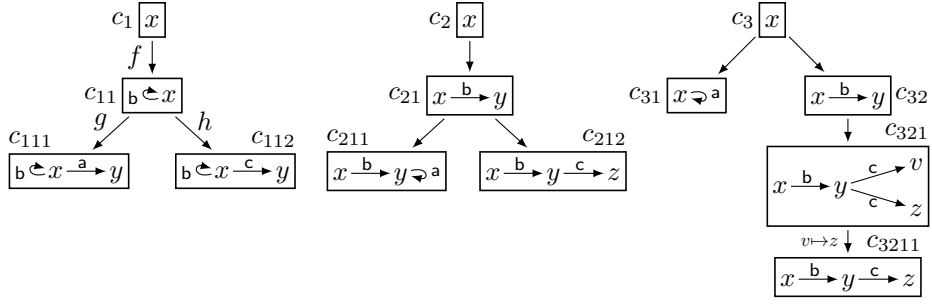
\begin{figure}[t]
\centering
\begin{tikzpicture}[on grid]
  \node[graph] (10) {\onenode{x}};
  \node[left=0 of 10.west,inner sep=0] {$c_1$};
  \node[graph,below=of 10] (11) {\oneloopleft{x}{b}}; 
  \node[left=0 of 11.west,inner sep=0] {$c_{11}$};
  \node[graph,below left=1 and 1.1 of 11] (111) {\looponeedge{x}{b}{a}{y}};
  \node[above right=0 of 111.north west,inner sep=1] {$c_{111}$};
  \node[graph,below right=1 and 1.1 of 11] (112) {\looponeedge{x}{b}{c}{y}};
  \node[above left=0 of 112.north east,inner sep=1] {$c_{112}$};
  
  \path (10) edge[->] node[left] {$f$} (11)
        (11) edge[->] node[above left,inner sep=1] {$g$} (111)
        (11) edge[->] node[above right,inner sep=1] {$h$}(112);

  \node[graph] (20) [right=4.2 of 10] {\onenode{x}};
  \node[left=0 of 20.west,inner sep=0] {$c_2$};
  \node[graph,below=of 20] (21) {\oneedge{x}{b}{y}}; 
  \node[left=0 of 21.west,inner sep=0] {$c_{21}$};
  \node[graph,below left=1 and 1.1 of 21] (211) {\oneedgeloop{x}{b}{y}{a}};
  \node[above right=0 of 211.north west,inner sep=1] {$c_{211}$};
  \node[graph,below right=1 and 1.1 of 21] (212) {\twoedge{x}{b}{y}{c}{z}};
  \node[above left=0 of 212.north east,inner sep=1] {$c_{212}$};
  
  \path (20) edge[->] (21)
        (21) edge[->] (211)
        (21) edge[->] (212);
  
  \node[graph,right=4 of 20] (30) {\onenode{x}};
  \node[left=0 of 30.west,inner sep=0] {$c_3$};
  \node[graph,below left=of 30] (31) {\oneloop{x}{a}}; 
  \node[left=0 of 31.west,inner sep=0] {$c_{31}$};
  \node[graph,below right=of 30] (32) {\oneedge{x}{b}{y}}; 
  \node[right=0 of 32.east,inner sep=0] {$c_{32}$};
  \node[graph,below=1.2 of 32] (321) {\onetwoedge{x}{b}{y}{c}{v}{c}{z}}; 
  \node[above left=0 of 321.north east,inner sep=1] {$c_{321}$};
  \node[graph,below=1.2 of 321] (3211) {\twoedge{x}{b}{y}{c}{z}}; 
  \node[above left=0 of 3211.north east,inner sep=1] {$c_{3211}$};

  \path (30) edge[->] (31)
        (30) edge[->] (32)
		(32) edge[->] (321)
		(321) edge[->] node[left] {\mapping{v&z}} (3211);
\end{tikzpicture}
\vspace*{-4mm}
\caption{Examples of conditions (see \cref{ex:conditions})}
\label{fig:conditions}
\end{figure}

\subsection{Satisfaction}

A condition expresses a property of arrows from its root to arbitrary objects. This is operationalised through the notion of \emph{satisfaction}.

\begin{definition}[satisfaction]\label{def:satisfaction}%
Let $g\of R\to G$ be an arbitrary arrow.
\begin{itemize}[defsep]
\item Let $c=(R,p_1\ccdots p_w)\in \AC R$. We say that \emph{$g$ satisfies $c$}, denoted $g\sat c$, if there is a branch $p_i$ such that $g\sat p_i$.

\item Let $p=(a\of R\to P,c)\in \AB R$. We say that \emph{$g$ satisfies $p$}, denoted $g\sat p$, if  there is an arrow $h\of P\to G$ such that
\begin{inumerate}
  \item $g=a;h$ and
  \item $h\nsat c$.\qed
\end{inumerate}
\end{itemize}
\end{definition}
\noindent
If $g\sat c$, we also say that $g$ is a \emph{model} for $c$. We call $p_i=(a_i,c_i)$ the \emph{responsible branch} and $h$ such that $g=a_i;h$ and $h\nsat c_i$ the \emph{witness} for $g\sat c$. Pictorially, $g\sat c$ with responsible branch $p_i$ and witness $h$ can be visualised as in \fcite{satisfaction}.

\begin{definition}[entailment, equivalence]\label{def:entailment}
For conditions $c_1,c_2\in \AC R$ over the same root $R$ we  define \emph{semantic entailment} $c_1\entails c_2$ and \emph{semantic equivalence} $c_1\equiv c_2$:
\begin{align*}
c_1 \entails c_2 & \text{ if for all arrows $g\of R\to G$, } g\sat c_1 \text{ implies } g\sat c_2 \\
c_1 \equiv c_2 & \text{ if for all arrows $g\of R\to G$, } g\sat c_1 \text{ if and only if } g\sat c_2 \enspace. 
\end{align*}
\raisedqed[12]
\end{definition}

\begin{example}\label{ex:satisfaction}
Let us consider some models for the conditions in \fcite{conditions}.
\begin{itemize}[nosep]
\item Let $G_1=\myinlinegraph{
\node (1) {$\bullet$};
\node (2) [right=of 1] {$\bullet$};
\path (1) edge[bend left=20,->] node[near start,above] {\lb} (2)
      (2) edge[bend left=20,->] node[near start,below] {\lb} (1)
	  (1) edge[loop left,->] node[left] {\lc} (1);
}$
and let $g$ be the morphism from the one-node discrete graph \inline{\onenode x} to the left hand node of $G_1$. It is clear that $g\nsat c_1$ because there is no witness for $c_{11}$ ($G_1$ does not have the required $\lb$-loop). Instead, we have $g\sat c_2$: the witness $h$ (for $c_{21}$) maps $y$ to the right hand node of $G_1$; and $h$ does not satisfy the subconditions of $c_{21}$ because it cannot be extended with either the $\la$-loop specified by $c_{211}$ or the outgoing $\lc$-edge specified by $c_{212}$. Similarly, $g\sat c_3$.

\item Let $G_2=\myinlinegraph{
\node (1) {$\bullet$};
\node (2) [right=of 1] {$\bullet$};
\node (3) [right=of 2] {$\bullet$};
\path (1) edge[bend left=20,->] node[near start,above] {\lb} (2)
      (2) edge[bend left=20,->] node[near start,below] {\lb} (1)
	  (1) edge[loop left,->] node[left] {\la} (1)
      (2) edge[->] node[above] {\lc} (3);
	  }$
and let $g,g',g''$ be the morphisms from \inline{\onenode x} to $G_2$ mapping \inline{\onenode x} to the left, mid and right node of $G_2$, respectively. None of these models satisfy $c_1$, for the same reason as above. Moreover, none satisfy $c_2$: though $g$ and $g'$ have witnesses for $c_{21}$, these are ruled out by either $c_{212}$ (in the case of $g$) or $c_{211}$ (in the case of $g'$); $g''$ does not even have a witness for $c_{21}$. Instead, both $g\sat c_3$ (in fact there are two distinct witnesses, one for $c_{31}$ and one for $c_{32}$) and $g'\sat c_3$ (due to $c_{32}$); but again $g''\nsat c_3$.

\item If $G_3=\inline{\oneloopleft \bullet b}$, then the only morphism $g$ from \inline{\onenode x} to $G_3$ has $g\sat c_1$, $g\sat c_2$ and $g\sat c_3$. 
\end{itemize}
In general, it can be checked that every model of $c_1$ is also a model of $c_2$ and every model of $c_2$ is a model of $c_3$, thus $c_1\entails c_2$ and $c_2 \entails c_3$. On the other hand, $c_2\nsat c_1$ as shown by $g:\inline{\onenode x}\to G_1$, and $c_3 \nsat c_2$ as shown by $g:\inline{\onenode x}\to G_2$.\qed
\end{example}

\subsection{Constructions over conditions}

We recall and introduce some special conditions and constructions that are used in the developments of this paper, together with their properties with respect to satisfaction.

\paragraph{Falsehood.}

For any graph $R$, there is an obvious ``trivial'' condition rooted in $R$, namely the one without any branches.
\begin{equation}\label{eq:bottom}
\bot_R \isdef (R,\epsilon) \enspace.
\end{equation}
The following result is immediate; it confirms that $\bot_R$ represents \emph{false} for models rooted in $R$.
\begin{proposition}[$\bot$ captures falsehood]\label{prop:bottom}
Let $R$ be a graph. Then $g\nsat \bot_R$ for any $g\of R\to G$.
\end{proposition}

\paragraph{Inversion.}

Given a condition over $R$, a straightforward construction is to ``push it down a level'' by making it a subcondition of the identity arrow $\id_R$.
\begin{definition}[inverse]\label{def:inverse}
Let $c\in \AC R$. The \emph{inverse} of $c$ is given by $\inv c\isdef (R,(\id_R,c))\in \AC R$.\qed
\end{definition}
\noindent
The following proposition states that inverse behaves like negation: a model satisfies a condition if and only if it does not satisfy its inverse. The proof is straightforward and omitted here.
\begin{proposition}[inverse captures negation]\label{prop:inverse}
Let $c\in \AC R$. For any $g\of R\to G$, $g\sat c$ iff $g\nsat \inv c$.
\end{proposition}

\paragraph{Shifting.}

We introduce two basic operations over conditions that allow one to change (\emph{shift}) the condition's root by exploiting arrows from or to the root. 
Given a condition $c \in \AC{R}$ and an incoming arrow $u\of S \to R$, we can \emph{upshift} $c$ along $u$ obtaining a condition $\ushift{c}{u} \in \AC{S}$. Dually, given an outgoing arrow $v\of R \to Q$, we can \emph{downshift} $c$ along $v$ by exploiting the existence of pushouts in $\cat{C}$, obtaining a condition $\dshift{c}{v} \in \AC{Q}$ (see \cref{fig:shift}).\footnote{In the literature (e.g., \cite{p:correct-GTSs-thesis,bchk:conditional-reactive-systems}), downshift is more commonly just called \emph{shift}.}

For the components of the pushout of two arrows $v\of R \to Q$ and $a \of R \to P$ we will use the notation illustrated in the following diagram, assuming that $( P\tdarr Q, v\tdarr a, a\tdarr v)$ is a concretely chosen pushout of $v$ and $a$ (which is defined up to isomorphism).

\begin{equation}\label{eq:pushout}
\vspace*{-2mm}\tikzscale{\begin{tikzpicture}[auto,node distance = 1.5cm,baseline=(a)]
      \node (R) {$R$};
      \node (Q) [right of=R] {$Q$};
      \node (P) [below of=R] {$P$};
    \node (PO) [below of=Q] {$P\tdarr Q$};

    \path
      (R) edge [->] node[auto] {$\scriptstyle{v}$}  (Q)
      edge [->] node (a) [auto,swap] {$\scriptstyle{a}$}  (P)
    (Q) edge [->] node[auto] {$\scriptstyle{a\tdarr v}$}  (PO)
    (P) edge [->] node[auto] {$\scriptstyle{v\tdarr a}$}  (PO); 
\pushoutmark{PO}{R}{0.3}{3mm}
\end{tikzpicture}}
\end{equation} 
  
\begin{figure}[t]
\centering
{  \begin{tikzpicture}[node distance = 2cm]
      \node (R) {$R$};
      \tri{R}{2.7}{1.2}{2.4}{}
      \node (S) [left of=R, xshift=-1.5cm] {$S$};
      \tri{S}{2.7}{1.2}{2.4}{}
      \node (Q) [right of=R, xshift=1.5cm] {$Q$};
      \tri{Q}{2.7}{1.2}{2.4}{}
      \node (Pi) [below of=R,yshift=4mm] {$P_i$};
      \tri{Pi}{1}{0.4}{0.8}{$\scriptstyle{b_i}$}
      \node (Pib) [below of=S,yshift=4mm] {$P_i$};
      \tri{Pib}{1}{0.4}{0.8}{$\scriptstyle{b_i}$}
     \node (Pis) [below of=Q,yshift=4mm] {$P_i \tdarr Q$};
     \tri{Pis}{1}{0.8}{1.6}{$\scriptstyle{\dshift{b_i}{(v\tdarr a_i)}}$}
     \node (lab) [above of=R-left, xshift=1mm, yshift=-5mm] {$c$};
     \node (labb) [above of=S-left, xshift=1mm, yshift=-5mm] {$\ushift{c}{u}$};
     \node (labs) [above of=Q-right, xshift=-1mm, yshift=-5mm] {$\dshift{c}{v}$};
     \pushoutmark{Pis}{R}{0.3}{3mm}
     
\path 
(S) edge [->] node[auto] {$\scriptstyle{u}$}  (R)
(R) edge [->] node[auto] {$\scriptstyle{v}$}  (Q)
     (S) edge [->] node[pos=.65, fill=black!5, inner sep=1pt] {$\scriptstyle{u;a_i}$}  (Pib)
    (R) edge [->] node[pos=.65, fill=black!5, inner sep=1pt] {$\scriptstyle{a_i}$}  (Pi)    
    (Q) edge [->] node[pos=.65, fill=black!5, inner sep=1pt] {$\scriptstyle{a_i\tdarr v}$}  (Pis)
    (Pi) edge [->] node[auto] {$\scriptstyle{v\tdarr a_i}$} (Pis)
    (S) edge [draw=none] node [pos=.55,xshift=2mm] {$\scriptstyle{\cdots}$} (S-left)
    edge [draw=none] node [pos=.55,xshift=-2mm] {$\scriptstyle{\cdots}$} (S-right)
    (R) edge [draw=none] node [pos=.6,xshift=2mm] {$\scriptstyle{\cdots}$} (R-left)
    edge [draw=none] node [pos=.6,xshift=-2mm] {$\scriptstyle{\cdots}$} (R-right)
    (Q) edge [draw=none] node [pos=.65,xshift=2mm] {$\scriptstyle{\cdots}$} (Q-left)
    edge [draw=none] node [pos=.65,xshift=-2mm] {$\scriptstyle{\cdots}$} (Q-right)
    ;
    
    \end{tikzpicture}}
\caption{Upshift (left) and downshift (right) of $c$ along $u$ resp.\ $v$.}
\label{fig:shift}
\end{figure}
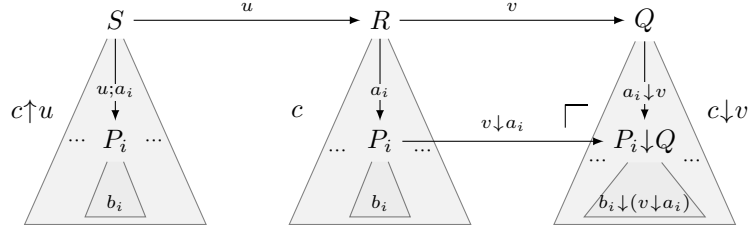

\begin{definition}[upshift]
  \label{def:upshift}
Let $u\of S \to R$ be an arrow in $\cat{C}$. 
\begin{itemize}[nosep]
\item For a condition $c=(R,p_1\ccdots p_w)\in \AC R$, the \emph{upshift of $c$ along $u$} is the condition $\ushift{c}{u} \in \AC{S}$ defined by $\ushift{c}{u} \isdef (S,\ushift{p_1}{u} \cdots \ushift{p_w}{u})$;
\item For a branch $p=(a,b)\in \AB R$, the \emph{upshift of $p$ along $u$} is the branch $\ushift p u\in \AB S $ defined by $\ushift p u\isdef (u;a,b)$.\qed
\end{itemize}
\end{definition}
\noindent
Note that upshifting $c$ along an incoming arrow $u$ is \emph{not} the same as constructing a branch $p=(u,c)$ where $c$ becomes the subcondition of $p$. The semantic difference will become clear below (\cref{prop:shiftSat}).
 
\begin{definition}[downshift]
  \label{def:downshift}
Let $v\of R \to Q$ be an arrow in $\cat C$.
\begin{itemize}[nosep]
\item For a condition $c=(R,p_1\ccdots p_w)\in \AC R$, the \emph{downshift of $c$ along $v$} is the condition $\dshift{c}{v} \in \AC{Q}$ defined by  $\dshift{c}{v} \isdef (Q, \dshift{p_1}{v}\cdots \dshift{p_w}{v})$.

\item For a branch $p=(a,b)\in \AB R$, the  \emph{downshift of $p$ along $v$}  is the branch $\dshift{p}{v} \in \AB{Q}$ defined by $\dshift p v\isdef (a \tdarr v, \dshift{b}{(v\tdarr {a})})$.\qed
\end{itemize}
\end{definition}
\noindent
Note that \emph{upshifting} a condition along an arrow only affects the top-level branches, while the branch conditions are preserved. Instead, \emph{downshifting} a condition along an arrow also affects all branch conditions, which are recursively shifted along the corresponding pushout arrows. 

\begin{example}\label{ex:shift}
Consider the following condition $c$ (which equals $c_2$ of \cref{fig:conditions}), with the arrows $u\of Q\to R$ and $v\of R\to S$ as in the following diagram:
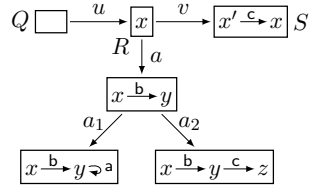
\[\vspace*{-4mm}\tikzscale{\begin{tikzpicture}[on grid,auto]
  \node[graph] (20) {\onenode{x}};
  \node[below left=0 of 20.south west,inner sep=-1] {$R$};
  \node[graph,below=1.2 of 20] (21) {\oneedge{x}{b}{y}}; 
  \node[graph,below left=1.2 and 1.2 of 21] (211) {\oneedgeloop{x}{b}{y}{a}};
  \node[graph,below right=1.2 and 1.2 of 21] (212) {\twoedge{x}{b}{y}{c}{z}};
  
  \path (20) edge[->] node {$a$} (21)
        (21) edge[->] node[above left,inner sep=0] {$a_1$} (211)
        (21) edge[->] node[above right,inner sep=0] {$a_2$} (212);

  \node[graph] (Q) [left=1.5 of 20] {$\emptygraph$};
  \node[left=0 of Q.west,inner sep=1] {$Q$};
  \node[graph] (S) [right=1.8 of 20] {$\oneedge{x'}cx$};
  \node[right=0 of S.east,inner sep=1] {$S$};

  \path (Q) edge[->] node[above] {$u$} (20)
        (20) edge[->] node[above] {$v$} (S);
\end{tikzpicture}}\]
The upward shifted condition $\ushift c u$ and downward shifted condition $\dshift c v$ are then given by:
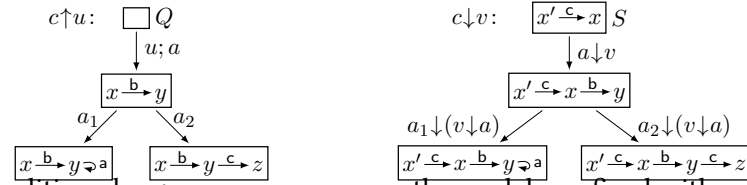
\[\vspace*{-2mm}\tikzscale{\begin{tikzpicture}[on grid,auto]
  \node[graph] (20) {\emptygraph};
  \node[right=0 of 20.east,inner sep=1] {$Q$};
  \node[left=1 of 20] {$\ushift c u\colon$};
  \node[graph,below=1.2 of 20] (21) {\oneedge{x}{b}{y}}; 
  \node[graph,below left=1.2 and 1.2 of 21] (211) {\oneedgeloop{x}{b}{y}{a}};
  \node[graph,below right=1.2 and 1.2 of 21] (212) {\twoedge{x}{b}{y}{c}{z}};
  
  \path (20) edge[->] node {$u;a$} (21)
        (21) edge[->] node[above left,inner sep=0] {$a_1$} (211)
        (21) edge[->] node[above right,inner sep=0] {$a_2$} (212);
\end{tikzpicture}}
\qquad\qquad
\tikzscale{\begin{tikzpicture}[on grid,auto]
  \node[graph] (20) {\oneedge{x'}cx};
  \node[right=0 of 20.east,inner sep=1] {$S$};
  \node[left=1.5 of 20] {$\dshift c v\colon$};
  \node[graph,below=1.2 of 20] (21) {\twoedge{x'}c{x}{b}{y}}; 
  \node[graph,below left=1.2 and 1.6 of 21] (211) {\twoedgeloop{x'}c{x}{b}{y}{a}};
  \node[graph,below right=1.2 and 1.6 of 21] (212) {\threeedge{x'}c{x}{b}{y}{c}{z}};
  
  \path (20) edge[->] node {$\dshift a v$} (21)
        (21) edge[->] node[left=.2] {$\dshift{a_1}{(\dshift va)}$} (211)
        (21) edge[->] node[right=.2] {$\dshift{a_2}{(\dshift va)}$} (212);
\end{tikzpicture}}
\]
\raisedqed
\end{example}
\noindent
Downshifting a condition along an arrow preserves the models prefixed with such arrow (again see \cite{p:correct-GTSs-thesis,bchk:conditional-reactive-systems}), whereas upshifting along an arrow preserves them, and also reflects them if the arrow is epi.

\begin{propositionE}[satisfaction for upshift and downshift]\label{prop:shiftSat}
Let $c \in \AC R$.
\begin{enumerate}[nosep]
  \item
  Let $v\of R \to S$. Then $v;g \sat c$ iff $g \sat \dshift{c}{v}$.
  \item 
  Let $u\of Q \to R$.  Then $g \entails \ushift{c}{u}$ iff there exists an arrow $g'$ such that $g = u;g'$ and $g' \entails c$.
  \item   
  If $u;g \sat \ushift{c}{u}$ and $u$ is epi then $g \sat c$.

\end{enumerate}
\end{propositionE}

\begin{proofE}~
\begin{enumerate}[defsep]
\item We proceed by induction on the depth of $c$, noting that downshift preserves the depth. If $\depth(c) = 0$ then $c = (R, \epsilon)$ and $\dshift{c}{v} = (Q,\epsilon)$; thus for every $g \of Q \to G$ we have $u;g \not \sat c$ and $g \not \sat \dshift{c}{u}$. Now let $\depth(c) = n > 0$ and assume that the statement holds for all conditions with depth smaller than $n$.

\begin{description}
\item[{$\Rightarrow$}] Let $g\of Q \to G$ and assume that $v;g \sat c$, meaning that there is a branch $p_i = (a_i\of R \to P_i, b_i)$ of $c$ with $i \in [1,w]$ and an $h: P_i \to G$ such that \tagone $a_i ; h = v;g$ and \tagtwo $h \not \sat b_i$. To show that $g \sat \dshift{c}{v}$, we can consider the branch $(a_i\tdarr v\of Q \to  P_i \tdarr Q, \dshift{b_i}{(v\tdarr a_i)})$ of $\dshift{c}{v}$, and as witness the arrow $h'\of P_i \tdarr Q \to G$ induced by the universal property of the pushout because of \tagone. In fact, we have $(a_i\tdarr v);h' = g$ by the pushout properties, while $h' \not \sat \dshift{b_i}{(v\tdarr a_i)}$ holds by contradiction. Indeed, by the induction hypothesis (since $\depth(b_i) < n$) we have $h' \sat \dshift{b_i}{(v\tdarr a_i)} \iff (v\tdarr a_i) ; h' \sat b_i$ which contradicts \tagtwo because $(v\tdarr a_i) ; h' = h$ by the pushout properties.
 
\item[{$\Leftarrow$}] Assume that $g \sat \dshift{c}{v}$. Thus there is a branch  
$(a_i\tdarr v\of Q \to  P_i \tdarr Q, \dshift{b_i}{(v\tdarr a_i)})$ of $\dshift{c}{v}$ and an $h'\of P_i \tdarr Q \to G$ such that $(a_i\tdarr v);h' = g$ and \tagthree $h' \not \sat \dshift{b_i}{(v\tdarr a_i)}$. We show that $v;g \sat c$ with responsible branch $(a_i\of R \to P_i,b_i)$ and witness \tagfour $h \isdef (v\tdarr a_i) ; h'$. In fact $a_i; h = a_i; v\tdarr a_i ; h' = v; (a_i\tdarr v); h' = v;g$. Moreover, by the induction hypothesis we can assume that $(v \tdarr a_i); h' \sat b_i \iff h' \sat \dshift{b_i}{(v \tdarr a_i)}$, therefore by \tagthree and \tagfour we conclude that $h \not \sat b_i$.     
\end{description}

\item
\begin{description}
\item[{$\Leftarrow$}] 
It is obvious that if $(a_i, b_i)$ is a responsible branch for $g' \sat c$ with witness $h\of P_i \to G$, then $(u;a_i,b_i)$ is a responsible branch for $g = u;g' \sat  \ushift{c}{u}$, with exactly the same witness. In fact $a_i; h = g'$ (and $h \nsat b_i$) implies $u;a_i ; h = u;g$ (and $h \nsat b_i$). \tagone If $u$ is epi also the converse is true, proving Point 3.

\item[$\Rightarrow$] Assume that $g \sat \ushift{c}{u}$, meaning that there is a branch $(u;a_i,b_i)$ of $\ushift{c}{u}$ with $i \in [1,w]$ and an $h\of P_i \to G$ such that $u;a_i ; h = g$. Let $g' \isdef a_i;h$: obviously $g' \sat c$ with responsible branch $(a_i, b_i)$ and witness $h$, since $g' = a_i ; h$, and $h \nsat b_i$.
\end{description}
\item See \tagone in the previous point. 
\end{enumerate}
\end{proofE}
\noindent
The notation chosen for the upshift operation is motivated by its symmetry with downshift. Upshift has a flavour of existential quantification, and indeed an operation denoted $\exists u \st c$ is introduced in~\cite{bchk:conditional-reactive-systems} having the same semantics, in the sense that both satisfy the second point of \cref{prop:shiftSat}.  
However, the notion of nested condition of~\cite{bchk:conditional-reactive-systems} and the concrete definition of the existential operation differ from ours: a precise comparison is deferred to future work. 

\paragraph{Join and meet.}

We define the \emph{join} and \emph{meet} of two conditions over the same root, and show that these capture the disjunction, respectively the conjunction in terms of satisfaction. The first of these, join, merely takes the union of the branches of the two conditions.

\begin{definition}[join]\label{def:join}
Let $c_i=\tupof{R,p^i_1\ccdots p^i_{w_i}} \in \AC R$ for $i=1,2$. The \emph{join of $c_1$ and $c_2$} is defined by $c_1\sqcup c_2\isdef (R,p^1_1\ccdots p^1_{w_1}\, p^2_1\ccdots p^2_{w_2})$.\qed
\end{definition}
\noindent
It is straightforward to show that the join captures disjunction.
\begin{proposition}[join captures disjunction]\label{prop:join}
Given two conditions $c_i \in \AC R$ for $i=1,2$ and an arrow $g\of R\to G$, $g\sat c_1\sqcup c_2$ if and only if either $g\sat c_1$ or $g\sat c_2$.
\end{proposition}
\noindent
We subsequently define the \emph{meet} of two conditions over the same root. The construction relies as an intermediate step on the meet of two branches.
\begin{definition}[meet]\label{def:meet}~
\begin{itemize}[nosep]
\item Let $p_i=(a_i,b_i)\in \AB R$ for $i=1,2$. The \emph{meet of $p_1$ and $p_2$} is a branch $p_1\sqcap p_2\in \AB R$, defined by $p_1\sqcap p_2 \isdef (a_1;(a_2\tdarr a_1),\dshift{b_1}{(a_2\tdarr a_1)}\sqcup \dshift{b_2}{(a_1\tdarr a_2)})$.

\item Let $c_i=(R_i,p^i_1\ccdots p^i_{w_i}) \in \AC R$ for $i=1,2$. The \emph{meet of $c_1$ and $c_2$} is a condition $c_1\sqcap c_2\in \AC R$, defined by $c_1\sqcap c_2\isdef \tupof{R,p_{1,1}\ccdots p_{w_1,1}\ccdots p_{1,w_2}\ccdots p_{w_1,w_2}}$, where $p_{m,n}=p^1_m\sqcap p^2_n$ for all $m\in [1,w_1], n\in [1,w_2]$.\qed
\end{itemize}
\end{definition}
\noindent
Note that, although the definition of the meet for branches might seem asymmetric due to the use of the arrow $a_1;(a_2\tdarr a_1)$, this is actually not the case as $a_1;(a_2\tdarr a_1)=a_2;(a_1\tdarr a_2)$ (see \eqref{eq:pushout}).

\begin{example}\label{ex:meet}
Since the meet is defined in terms of downshift, one simple case is already given by a reinterpretation of the downshift example in \cref{ex:shift}: the result of that downshift can be seen as the meet of the condition $c$ and the condition $(R,(v,(S,\epsilon)))$.

For a slightly more elaborate example, consider the following two conditions $c_1,c_2\in\AC{\inline{\twonode x y}}$ (with $\FOL$ interpretations given below the condition):
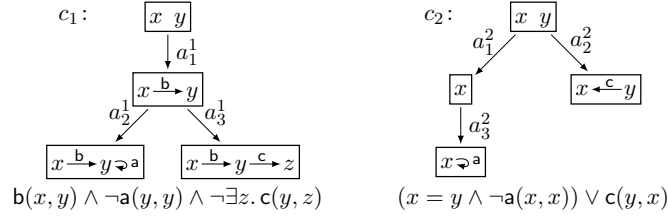
\[\tikzscale{\begin{tikzpicture}[on grid,auto]
  \node[graph] (20) {\twonode x y};
  \node [left=1.5 of 20] {$c_1\colon$};
  \node[graph,below=1.2 of 20] (21) {\oneedge{x}{b}{y}}; 
  \node[graph,below left=1.2 and 1.2 of 21] (211) {\oneedgeloop{x}{b}{y}{a}};
  \node[graph,below right=1.2 and 1.2 of 21] (212) {\twoedge{x}{b}{y}{c}{z}};
  
  \path (20) edge[->] node {$a^1_1$} (21)
        (21) edge[->] node[above left,inner sep=0] {$a^1_2$} (211)
        (21) edge[->] node[above right,inner sep=0] {$a^1_3$} (212);
        
  \node [below=3 of 20]
  {$\lb(x,y)\wedge \neg\la(y,y)\wedge \neg\exists z\st \lc(y,z)$};
\end{tikzpicture}}
\qquad
\tikzscale{\begin{tikzpicture}[on grid,auto]
  \node[graph] (20) {\twonode x y};
  \node [left=1.5 of 20] {$c_2\colon$};
  \node[graph,below left=1.2 and 1.2 of 20] (21) {\onenode x}; 
  \node[graph,below right=1.2 and 1.2 of 20] (22) {\onebackedge x c y}; 
  \node[graph,below=1.2 and 1.2 of 21] (211) {\oneloop{x}{a}};
  
  \path (20) edge[->] node[above left,inner sep=0] {$a^2_1$} (21)
             edge[->] node[above right,inner sep=0] {$a^2_2$} (22)
        (21) edge[->] node {$a^2_3$} (211);
        
  \node [below=3 of 20]
  {$(x=y\wedge \neg\la(x,x)) \vee \lc(y,x)$};
\end{tikzpicture}}\vspace*{-1mm}
\]
In this case, the meet can be depicted as
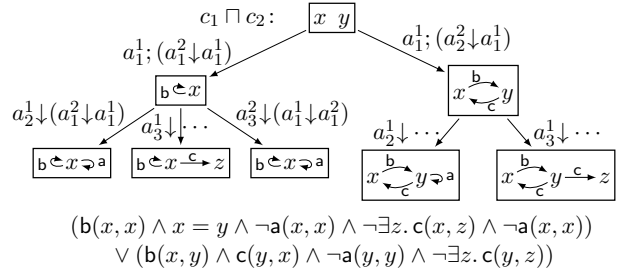
\[\tikzscale{\begin{tikzpicture}[on grid]
  \node[graph] (20) {\twonode x y};
  \node [left=1.5 of 20] {$c_1\meet c_2\colon$};
  \node[graph,below left=1.2 and 2.5 of 20] (21) {\oneloopleft x b}; 
  \node[graph,below left=1.2 and 1.8 of 21] (211) {\twoloop x b a};
  \node[graph,below=1.2 of 21] (212) {\looponeedge x b c z};
  \node[graph,below right=1.2 and 1.8 of 21] (213) {\twoloop x b a};
  \node[graph,below right=1.2 and 2.5 of 20] (22) {\backforthedge{x}{b}{y}{c}};
  \node[graph,below left=1.4 and 1.2 of 22] (221) {\backforthedgeloop{x}{b}{y}{c} a};
  \node[graph,below right=1.4 and 1.2 of 22] (222) {\backforthedgeedge{x}{b}{y}{c} c z};
  
  \path (20) edge[->] node[above left,inner sep=0,near end] {$a^1_1;(a^2_1\tdarr a^1_1)$} (21)
             edge[->] node[above right,inner sep=0] {$a^1_1;(a^2_2\tdarr a^1_1)$} (22)
        (21) edge[->] node[above left,inner sep=0] {$a^1_2\tdarr(a^2_1\tdarr a^1_1)$} (211)
             edge[->] node {$a^1_3\tdarr\cdots\ \ $} (212)
             edge[->] node[above right,inner sep=0] {$a^2_3\tdarr(a^1_1\tdarr a^2_1)$} (213)
        (22) edge[->] node[left,inner sep=0] {$a^1_2\tdarr\cdots\ $} (221)
             edge[->] node[right] {$\ a^1_3\tdarr\cdots$} (222);
        
  \node [below=3.7 of 20]
  {$\begin{array}{c}
    (\lb(x,x) \wedge x=y\wedge \neg \la(x,x)\wedge \neg\exists z\st \lc(x,z) \wedge \neg \la(x,x)) \\
    {}\vee (\lb(x,y)\wedge \lc(y,x) \wedge \neg\la(y,y) \wedge \neg\exists z\st \lc(y,z))
    \end{array}$};
\end{tikzpicture}}\]\raisedqed
\end{example}
\noindent
From the example, the intuition can be gathered that the construction of the meet is on some level related with the distributivity of logical conjunction over disjunction; and in fact, we can show that $(\bigsqcup_{i\in I}c_i)\sqcap (\bigsqcup_{j\in J}c'_j)\cong \bigsqcup_{i\in I,j\in J} (c_i\sqcap c'_j)$, where $\cong$ stands for isomorphism between conditions.

We show that the meet indeed captures conjunction, as announced. 
\begin{propositionE}[meet captures conjunction]\label{prop:meet} Let $g\of R\to G$ be an arbitrary arrow.
\begin{enumerate}[nosep]
\item If $p_i\in \AB R$ for $i=1,2$, then $g\sat p_1\sqcap p_2$ if and only if $g\sat p_1$ and $g\sat p_2$.

\item If $c_i\in \AC R$ for $i=1,2$, then $g\sat c_1\sqcap c_2$ if and only if $g\sat c_1$ and $g\sat c_2$.
\end{enumerate}
\end{propositionE}
\begin{proofE}~
\begin{enumerate}[defsep]
\item Assume $p_i=(a_i\of R\to (P_i,b_i))$ for $i=1,2$, let $P=P_1\tdarr P_2$, and let $b_1'=\dshift{b_1}{(a_2\tdarr a_1)}$ and $b_2'=\dshift{b_2}{(a_1\tdarr a_2)}$. Hence $b_1',b_2'\in \AC P$ and $p_1\sqcap p_2=(a_1;(a_2\tdarr a_1),b_1'\sqcup b_2')$.

\begin{description}
\item[If.] Assume $g\sat p_i$ for $i=1,2$; then there are $h_i\of P_i\to G$ such that $g=a_i;h_i$ and $h_i\nsat b_i$.
 
Because of the pushout construction, there is a unique $h\of P\to G$ such that $(a_2\tdarr a_1);h=h_1$ and $(a_1\tdarr a_2);h=h_2$; hence $a_1;(a_2\tdarr a_1);h=a_2;(a_1\tdarr a_2);h=g$. If $h\sat b_1'\sqcup b_2'$, due to \cref{prop:join} we have $h\sat b_1'$ or $h\sat b_2'$. Let us assume the former; the other case is symmetrical. \Cref{prop:shiftSat} then implies $(a_2\tdarr a_1);h\sat b_1$, which contradicts $h_1\nsat b_1$. We may conclude $g\sat p_1\sqcap p_2$.

\item[Only if.] Assume $g\sat p_1\sqcap p_2$; then there is a $h\of P\to G$ such that $g=a_1;(a_2\tdarr a_1);h$ and $h\nsat b_1'\sqcup b_2'$, meaning (due to \cref{prop:join}) $h\nsat b_1'$ and $h\nsat b_2'$. Let $h_1=(a_2\tdarr a_1);h$; then clearly $g=a_1;h_1$, and Prop.~\ref{prop:shiftSat} implies $h_1\nsat b_1$. It follows that $g\sat p_1$. A symmetrical argument shows that also $g\sat p_2$.
\end{description}

\item Assume $c_i=\tupof{R,p^i_1\ccdots p^i_{w_i}}$, and let $p_{j,k}=p^1_j\sqcap p^2_k$ for $j\in[1,w_1],k\in [1,w_2]$.

\smallskip
$g\sat c_i$ for $i=1,2$ if and only if there are $j\in[1,w_1], k\in[1,w_2]$ such that $g\sat p^1_j$ and $g\sat p^2_k$. According to Clause~1 above, this is the case if and only if there are $j\in[1,w_1], k\in[1,w_2]$ such that $g\sat p_{j,k}$, which in turn is equivalent to $g\sat c_1\sqcap c_2$.
\end{enumerate}
\end{proofE}
\noindent
One may wonder how useful the meet is, given that there is obviously another way to capture conjunction, due to De Morgan duality, through a combination of the (technically much simpler) inversion and join constructions. Indeed, \cref{prop:inverse,prop:join} together guarantee that, for any $c_1,c_2\in\AC R$, if we let
$c= \inv{(\inv{c_1}\sqcup \inv{c_2})}$, then $g\sat c$ if and only if $g\sat c_1$ and $g\sat c_2$; in other words, $c\equiv c_1\sqcap c_2$. However, those two constructions give rise to very different conditions, as the following simple example illustrates.

\begin{example}\label{ex:meet vs De Morgan}
Consider $a_1 = \inline{\onenode x}\to\inline{\oneloopleft x a}$ and $a_2=\inline{\onenode x}\to \inline{\oneedge x b y}$, and let $c_i=c_{a_i}$ for $i=1,2$. The following depicts the two equivalent conditions obtained taking the meet of $c_1$ and $c_2$ (left hand side) and applying the De Morgan duality (right hand side).

\vspace*{-2mm}\[\tikzscale
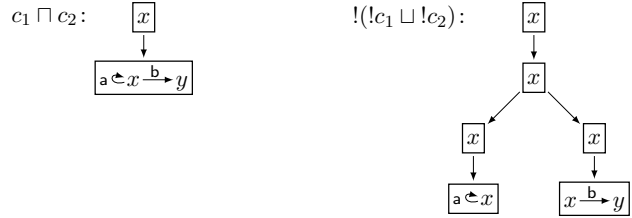
{\begin{tikzpicture}[on grid,baseline=(R)]
\node[graph] (R) {$\onenode x$};
\node [left=1.5 of R] {$c_1\meet c_2\colon$};
\node[graph] (P) [below=of R] {$\looponeedge x a b y$};
\path (R) edge[->] (P);
\end{tikzpicture}
\qquad\qquad\qquad
\begin{tikzpicture}[on grid,baseline=(R)]
\node[graph] (R) {$\onenode x$};
\node [left=2 of R] {$\inv{(\inv{c_1}\join \inv{c_2})}\colon$};
\node[graph] (R1) [below=of R] {$\onenode x$};
\node[graph] (R2) [below left=of R1] {$\onenode x$};
\node[graph] (P2) [below=of R2] {$\oneloopleft x a$};
\node[graph] (R3) [below right=of R1] {$\onenode x$};
\node[graph] (P3) [below=of R3] {$\oneedge x b y$};
\path (R) edge[->] (R1)
      (R1) edge[->] (R2)
	       edge[->] (R3)
      (R2) edge[->] (P2)
      (R3) edge[->] (P3);
\end{tikzpicture}}\]
Clearly, for this particular example, the left hand condition is more concise and understandable than the right hand condition. In general, there is a trade-off: meet ($c_1\meet c_2$) creates a condition of which the depth is the \emph{maximum} of $\depth(c_1)$ and $\depth(c_2)$, whereas the branching width is the \emph{product} of $|c_1|$ and $|c_2|$, and for each of the new branches $p=(a,b)$ (which always equal $p^1\sqcap p^2$ for some branch $p^1=(a^1,b^1)$ of $c_1$ and some branch $p^2=(a^2,b^2)$ of $c_2$) the next-level widths \emph{sum up}: $|b|=|b^1|+|b^2|$. For the De Morgan-based construction ($\inv{(\inv{c_1}\join \inv{c_2})}$), on the other hand, the depth is always the maximum of $\depth(c_1)$ and $\depth(c_2)$ \emph{increased by 2}, whereas the width is merely the \emph{sum} of $|c_1|$ and $|c_2|$ (and the branch subconditions remain unchanged).\qed
\end{example}
\noindent
The following proposition shows that the constructions presented in this section are in a sense complete: they can be used to build any condition, using arbitrary arrows as upward-shift operands.
\begin{propositionE}\label{prop:construction}
Any condition $c$ can be constructed through a combination of constants $\bot_R$, inversion {\upshape ($\inv{}$)}, upshift {\upshape($\tuarr$)} and join {\upshape($\sqcup$)}.
\end{propositionE}

\begin{proofE}
We inductively define a family of transformations $t_R\of \AC R\to \AC R$ such that $t_R(c)=c$ for all $c\in\AC R$, using only the constructions allowed by the proposition. For conditions $c=(R,p_1\ccdots p_w)\in \AC R$ and branches $p_i=(a_i,b_i)$ with $a_i\of R\to P_i$ for $i\in[1,w]$, let 
\begin{align*}
t_R\colon c
   & \mapsto
     \begin{cases}
       \bot_R & \text{if } w=0 \\ 
       (\inv{t_{P_1}(b_1)}\:\tuarr\: a_1) \join t_R((R,p_2\ccdots p_w)) & \text{otherwise.} 
     \end{cases}
\end{align*}
The proof of $t_R(c)=c$ is by induction on the depth of $c\in \AC R$. The base case is due to $w=0$ if $\depth(c)=0$, and hence $t_R(c)=\bot_R=(R,\epsilon)=c$. The induction step is based on inner induction on the width $w$, using that (by the outer induction hypothesis) $t_{P_1}(b_1)=b_1$. As an intermediate step note that, if $p=(a,b)$ with $a\of R\to P$ and $b\in \AC P$ such that $t_P(b)=b$ then
\[ \inv{t_P(b)}\:\tuarr\: a = \inv b\:\tuarr a = (P,(\id_P,b))\tuarr a=(R,(a;\id_P,b)) = (R,p) \enspace. \]
The base case and induction step for the inner induction, are, respectively
\begin{align*}
t_R(c) = \begin{cases}
  \bot_R = c & \text{if }w=0 \\
  (\inv{t_{P_1}(b_1)}\:\tuarr\: a_1) \join t_R((R,p_2\ccdots p_w)) \\
  \hphantom{\bot_R} = (R,p_1) \join (R,p_2\ccdots p_w) \\
  \hphantom{\bot_R} = (R,p_1\ccdots p_w) = c & \text{otherwise.}
  \end{cases}
\end{align*}
\end{proofE}

\section{Morphisms}
\label{sec:morphisms}

The notion of \emph{morphism} of nested conditions has not received much attention in the literature. When considered at all, such as in \cite{bchk:conditional-reactive-systems,sksclo:coinductive-techniques-for-satisfiability}, morphisms are essentially based on the semantics in terms of satisfaction. Indeed, entailment (see \cref{def:entailment}) establishes a preorder on conditions over the same root: 
given conditions $c_1,c_2 \in \AB{R}$ we have $c_1 \leq c_2$ iff $c_1 \entails c_2$. 
The question that we address in this section is to establish a meaningful \emph{structural} notion of condition morphism. That is, given the fact that a condition is essentially a diagram in the category $\bC$, a structural morphism from $c_1$ to $c_2$ consists of arrows between objects of $c_1$ and $c_2$ (when viewed as diagrams),  satisfying certain conditions on the resulting subdiagrams. 

A morphism will relate the roots of the source and target conditions with an arrow, and will fix a relation between the branches of the source and target conditions. For each pair of related branches the morphism will include a set of ``next-level'' morphisms between the corresponding branch conditions, which will go the other way around (i.e., from the target branch condition to the source branch condition): this is motivated by our definition of satisfaction, which flips the direction of the relation at every subsequent nesting level. This general notion of structural morphisms will be specialized in the next section to \emph{reflective} and \emph{preservative} morphisms, which will be shown to witness entailment in the expected direction.

\begin{definition}[condition morphism]\label{def:morphism}
Let $c_i = (R_i,(a^i_1,b^i_1)\ccdots (a^i_{w_i}, b^i_{w_i}))\in \AC{R_i}$ for $i=1,2$. A \emph{condition morphism} from $c_1$ to $c_2$ is a tuple $m=(v,\Nrel,N)\of c_1 \to c_2$ such that

\begin{itemize}[nosep]
\item $v\of R_2\to R_1$ is an arrow, called the \emph{root arrow};
\item $\Nrel \subseteq [1, w_1] \times [1, w_2]$ is a relation between the branches of $c_1$ and $c_2$;
\item $N$ is a family of non-empty sets of condition morphisms indexed by $\Nrel$, that is $N = \setof{ N_{jk} \mid (j,k) \in \Nrel}$, such that $n:b^2_k \to b^1_j$ for each $n \in N_{jk}$.
\end{itemize}
The components of a condition morphism $m$ are sometimes denoted $(v_m,\Nrel_m,N_m)$.\qed
\end{definition}
\noindent
Note that the root arrow $v_m$ of a morphism $m\of c_1\to c_2$ goes in the \emph{opposite} direction of $m$, namely from the root of $c_2$ to the root of $c_1$. In general, there are no requirements on the relation $\Nrel$ of $m$.
Moreover, for every $n\in N_{jk}$, the following $n$-based subdiagram emerges, which for now does not have to satisfy any constraints --- in particular, it does not have to commute.

\begin{tikzequation}\label{eq:subsquare}
\tikzscale{\begin{tikzpicture}[on grid,node distance = 1.5cm,baseline=(a1)]
\node (R1) {$R_1$};
\node (P1) [below=of R1] {$P^1_j$};
\node (R2) [right=of R1] {$R_2$};
\node (P2) [below=of R2] {$P^2_k$};

\path
  (R1) edge [->] node (a1) [left] {$a^1_j$} (P1)
  (R2) edge [->] node [auto] {$a^2_k$} (P2)
       edge [->] node [above] {$v$} (R1)
  (P1) edge [->] node [auto] {$v_n$} (P2);

\node [left=3 of a1] {$n$-based subsquare};
\end{tikzpicture}}
\end{tikzequation}%
However, in practice we restrict to two subclasses of morphisms that satisfy further properties (defined in \cref{sec:reflective and preservative morphisms} below), both on $\Nrel$ and on the induced subsquares of \eqref{eq:subsquare}. First, however, we introduce identities and composition of morphisms and show their expected properties.
\begin{definition}[identity morphisms and composition]\label{def:morphism composition}
Let $c=(R,(a_1,b_1)\ccdots (a_{w}, b_{w})) \in \AC{R}$ be a condition; then $\id_c\of c\to c$ is defined by
\begin{equation}
\id_{c}
 = (\id_{R},\Id_{[1,w]},\setof{N_{jj} \isdef \setof{\id_{b_j}} \mid (j,j) \in \Id_{[1,w]}}) 
     \label{eq:id condition morphism} \enspace.
\end{equation}
Let $c_i = (R_i,(a^i_1,b^i_1)\ccdots (a^i_{w_i}, b^i_{w_i}))\in \AC{R_i}$ for $i=1,2,3$ and let $m_i=(v_i,\Nrel_i,N^i) \of c_i\to c_{i+1}$ be condition morphisms for $i=1,2$. Then $m_1;m_2\of c_1\to c_3$ is defined by
\begin{flalign} \label{eq:condition morphism composition}
& m_1;m_2 = (v_2;v_1,\Nrel_1;\Nrel_2, \{N_{ik} \isdef \{m;n \of b^3_k \to b^1_i \mid & \\ 
 & \qquad \qquad \exists j \in [1,w_2] \st m\of b^3_k \to b^2_j \in N^2_{jk} \land n\of b^2_j \to b^1_i \in N^1_{ij}\} \mid (i,k) \in \Nrel_1 ; \Nrel_2\})
  \enspace.& \notag 
\end{flalign}
\raisedqed[14]
\end{definition}
\begin{proposition}\label{prop:morphism composition}
Condition morphism composition is associative, with left and right identities $\id_c$.
\end{proposition}  

To conclude this subsection, observe that the potential morphisms between two given conditions are finite, even if their number can be very large. In principle it is possible to design an algorithm that enumerates all potential morphisms between two conditions. Also, by exploting the results of the next section, some of those morphisms could witness the fact that one condition entails the other. However, as discussed in the concluding section we don't expect that this approach can be competitive with existing methods for checking entailment. 

\subsection{Reflective and preservative morphisms}
\label{sec:reflective and preservative morphisms}

We specialize here the general notion of morphism in order to capture entailment between conditions. \emph{Preservative} morphisms will witness entailment from the source to the target condition, and \emph{reflective} ones similarly, but in the other direction. We obtain this by constraining the relations between the branches of the conditions and by imposing certain properties on the resulting squares formed by the arrows of the morphism and the branches of the conditions, as in \eqref{eq:subsquare}. The properties are defined in terms of \emph{reflection squares} and \emph{preservation squares}, which are defined below. Intuitively, such properties ensure that the morphism preserves the satisfaction of branches in the expected direction.

\begin{definition}[reflection square]\label{def:reflection square}
Let $a_1\of R_1\to P_1$, $a_2\of R_2\to P_2$, $v\of R_2\to R_1$ and $v'\of P_1\to P_2$, and consider the (not necessarily commuting) diagram on the left:
\begin{tikzequation}\label{eq:reflection square}
\tikzscale{\begin{tikzpicture}[on grid,node distance = 1.5cm,baseline=(a1)]
\node (R1) {$R_1$};
\node (P1) [below=of R1] {$P_1$};
\node (R2) [right=of R1] {$R_2$};
\node (P2) [below=of R2] {$P_2$};

\path
  (R1) edge [->] node (a1) [auto] {$a_1$} (P1)
  (R2) edge [->] node [auto] {$a_2$} (P2)
       edge [->] node [above] {$v$} (R1)
  (P1) edge [->] node [auto] {$v'$} (P2);
\end{tikzpicture}}
\qquad \qquad
\tikzscale{\begin{tikzpicture}[on grid,node distance = 1.5cm,baseline=(a1)]
\node (R1) {$R_1$}; 
\node (P1) [below=of R1] {$P_1$};
\node (R2) [right=of R1] {$R_2$};
\node (P2) [below=of R2] {$P_2$};

\path
  (R1) edge [->] node (a1) [auto] {$a_1$} (P1)
  (R2) edge [->] node [auto] {$a_2$} (P2)
       edge [->] node [above] {$v$} (R1)
  (P1) edge [->] node [auto] {$v'$} (P2);

\node (G) [below left=of P1] {$G'$};
\path
  (R1) edge [->,bend right=25] node[left] {$g'$} (G)
  (P2) edge [->,bend left=25] node[auto,inner sep=1,pos=.3] {$h'$} (G);
  
\path (G) edge[white] node[pos=.7,black] {$\circlearrowright$} (P1);
\end{tikzpicture}}
\end{tikzequation}%
The diagram is called a \emph{reflection square} if for any pair $g',h'$ that makes the outer subdiagram to the right commute (i.e., such that $v;g'=a_2;h'$), its bottom left subdiagram commutes (i.e., $g'=a_1;v';h'$).\qed
\end{definition}

\noindent
Therefore a reflection square guarantees that if the branch $a_2$ satisfies $v;g'$  with witness $h'$, then $a_1$ potentially satisfies $g'$ with witness $v';h'$. The dual to a reflection square is a \emph{preservation square}, which ensures that if $a_1$ satisfies $g$ with witness $h$ then $a_2$ can potentially satisfy $v;g$ with some witness $h'$ such that $h=v';h'$:

\begin{definition}[preservation square]\label{def:preservation square}
Let $a_1\of R_1\to P_1$, $a_2\of R_2\to P_2$, $v\of R_2\to R_1$ and $v'\of P_1\to P_2$, and consider the (not necessarily commuting) diagram on the left:
\begin{tikzequation}\label{eq:preservation square}
\tikzscale{\begin{tikzpicture}[on grid,node distance = 1.5cm,baseline=(a1)]
\node (R1) {$R_1$};
\node (P1) [below= of R1] {$P_1$};
\node (R2) [right= of R1] {$R_2$};
\node (P2) [below= of R2] {$P_2$};

\path
  (R1) edge [->] node (a1) [auto] {$a_1$} (P1)
  (R2) edge [->] node [auto] {$a_2$} (P2)
       edge [->] node [above] {$v$} (R1)
  (P1) edge [->] node [auto] {$v'$} (P2);
\end{tikzpicture}}
\qquad \qquad
\tikzscale{\begin{tikzpicture}[on grid,node distance = 1.5cm,baseline=(a1)]
\node (R1) {$R_1$};
\node (P1) [below=of R1] {$P_1$};
\node (R2) [right=of R1] {$R_2$};
\node (P2) [below=of R2] {$P_2$};

\path
  (R1) edge [->] node (a1) [auto] {$a_1$} (P1)
  (R2) edge [->] node [auto] {$a_2$} (P2)
       edge [->] node [above] {$v$} (R1)
  (P1) edge [->] node [auto] {$v'$} (P2);

\node (G) [below left=of P1] {$G$};
\path (R1) edge[->,bend right=25] node[auto] (g) [left] {$g$} (G)
      (P1) edge[->] node[below right,inner sep=2,pos=.4] {$h$} (G)
      (P2) edge[dashed,->,bend left=25] node[auto,inner sep=1,pos=.3] (h') {$h'$} (G);
      
\path (P1) edge[white] node[black,pos=.4] {$\circlearrowright$} (g)
           edge[white] node[black,pos=.4] {$\circlearrowright$} (h');
\end{tikzpicture}}
\end{tikzequation}%
The diagram is called a \emph{preservation square} if for any pair of arrows $g,h$ such that $g=a_1;h$, there is an arrow $h'$ such that $h=v';h'$ and $v;g=a_2;h'$.\qed
\end{definition} 

\begin{example}[reflection and preservation squares]\label{ex:square}
Consider the following (non-commuting) diagrams:
\begin{tikzequation}\label{eq:ex-squares}
\tikzscale{\begin{tikzpicture}[on grid,auto,node distance = 1.5cm,baseline=(a1)]
\node[graph] (R1) {\twonode x z};
\node[left=0 of R1.west,inner sep=1] {$R_1$};
\node[graph,below=of R1] (P1) {\onenode z};
\node[left=0 of P1.west,inner sep=1] {$P_1$};
\node[graph,right=2.5 of R1] (R2) {\threenode x y z}; 
\node[right=0 of R2.east,inner sep=1] {$R_2$};
\node[graph,below=of R2] (P2) {\twonode y z};
\node[right=0 of P2.east,inner sep=1] {$P_2$};
  
\path
   (R1) edge[->] node[left] (a1) {$a_1$} node[right] {$\mapping{x&z}$} (P1)
   (R2) edge[->] node[left] {$a_2$} node[right] {$\mapping{x&y}$} (P2)
        edge[->] node[above] {$v$} node[below,inner sep=0] {$\mapping{y&z}$} (R1)
   (P1) edge[->] node {$v'$} (P2);
\end{tikzpicture}}
\qquad
\tikzscale{\begin{tikzpicture}[on grid,auto,node distance = 1.5cm,baseline=(a)]
\node[graph] (R1) {\onenode x};
\node[left=0 of R1.west,inner sep=1] {$R_1$};
\node[graph,below=of R1] (P1) {\oneloopleft x a};
\node[left=0 of P1.west,inner sep=1] {$P_1$};
\node[graph,right=3 of R1] (R2) {\twonode x y}; 
\node[below left=0 of R2.south west,inner sep=-2] {$R_2$};
\node[graph,below=of R2] (P2) {\looponeedge x a a y};
\node[right=0 of P2.east,inner sep=1] {$P_2$};
  
\path
   (R1) edge[->] node (a) {$a_1$} (P1)
   (R2) edge[->] node {$a_2$} (P2)
        edge[->] node[above] {$v$} node[below,inner sep=0] {$\mapping{y&x}$} (R1)
   (P1) edge[->] node {$v'$} (P2);
\end{tikzpicture}}
\qquad
\tikzscale{\begin{tikzpicture}[on grid,auto,node distance = 1.5cm,baseline=(a)]
\node[graph] (R1) {\nodeedge z x a y};
\node[left=0 of R1.west,inner sep=1] {$R_1$};
\node[graph,below=of R1] (P1) {\oneedge x a y};
\node[left=0 of P1.west,inner sep=1] {$P_1$};
\node[graph,right=3 of R1] (R2) {\oneedge x a y}; 
\node[right=0 of R2.east,inner sep=1] {$R_2$};
\node[graph,below=of R2] (P2) {\spangraph z a x a y};
\node[right=0 of P2.east,inner sep=1] {$P_2$};
  
\path
   (R1) edge[->] node[left] (a) {$a_1$} node[right] {$\mapping{z&x}$} (P1)
   (R2) edge[->] node[right] {$a_2$} node[left] {$\mapping{y&z}$} (P2)
        edge[->] node[above] {$v$} (R1)
   (P1) edge[->] node {$v'$} (P2);
\end{tikzpicture}}
\end{tikzequation}%
The left hand side and middle diagrams are reflection squares, whereas the right hand side is a preservation square. (This will follow from \cref{prop:reflection square,prop:preservation square} below; see \cref{ex:preservation square,ex:reflection square}).
\end{example}

\noindent
The underlying principle for a morphism $m$ from $c_1$ to $c_2$ that reflects satisfaction is that it must identify, for every $c_2$-branch $p^2_i$, a $c_1$-branch $p^1_j$ that it entails. Therefore we require the relation between branches to induce a function from $c_2$-branches to $c_1$-branches. Additionally, the subsquares formed by the root arrows and the branches of $c_1$ and $c_2$ (see \eqref{eq:subsquare}) must be reflection squares. Dually, for a morphism that preserves satisfaction, we require a function from $c_1$-branches to $c_2$-branches, and the subsquares should be preservation squares.

\begin{definition}[reflective and preservative morphisms]\label{def:reflective-preservative}~
Let $c_i = (R_i,p^i_1\ccdots p^i_{w_i})\in \AC{R_i}$ for $i=1,2$, and let $m=(v,\Nrel,N)\of c_1\to c_2$ be a condition morphism.
\begin{itemize}[nosep]
\item $m$ is called \emph{reflective} if the relation $\Nrel \subseteq [1,w_1]\times [1,w_2]$ is a function from $c_2$-branches to $c_1$-branches (i.e., $\forall k \in [1,w_2]\st \exists!\; j \in [1,w_1]\st (j,k)\in \Nrel$), and for all $(j,k) \in \Nrel$  it holds that $N_{jk}=\setof{n}$ for some reflection morphism $n$ such that the $n$-based subsquare (see \eqref{eq:subsquare}) is a reflection square; 
\item $m$ is called \emph{preservative} if relation $\Nrel$ is a function from $c_1$-branches to $c_2$-branches (i.e., $\forall j \in [1,w_1]\st \exists!\; k \in [1,w_2]\st (j,k)\in \Nrel$), and for all $(j,k) \in \Nrel$ it holds that $N_{jk}=\setof{n}$  for some preservation morphism $n$ such that the $n$-based subsquare is a preservation square.\qed 
\end{itemize}
\end{definition}
\noindent
Note that if $\depth(c_2)=0$ [resp.~$\depth(c_1)=0$] for the target [resp.\ source] of some reflective [preservative] condition morphism $m\of c_1\to c_2$, then $N=\varnothing$, hence the requirement on the family $N$ is vacuously fulfilled. We will show later in \cref{prop:reflective-preservative composition} that reflective and preservative morphisms compose. 

We now come to the core property of reflective and preservative morphisms, namely that they guarantee a one-way correspondence between the models satisfying their source and target conditions: from target to source for reflective morphisms, and from source to target for preservative ones.

\begin{propositionE}[reflective morphisms reflect satisfaction]\label{prop:reflective}
Let $c_i\in \AC{R_i}$ be conditions for $i=1,2$. If $m:c_1\to c_2$ is a  reflective condition morphism, then $v_m;g\sat c_2$ implies $g \sat c_1$ for all arrows $g\of R_1\to G$. 
\end{propositionE}

\begin{proofE}
By induction on $\depth(c_1)+\depth(c_2)$. Let $c_i=(R_i,p^i_1\cdots p^i_{w_i})$ such that $p^i_j=(a^i_j,b^i_j)$ for $i=1,2$ and all $j\in[1,w_i]$, and let $m=(v_m,\Nrel_m,N_m)$.
\begin{description}
\item[Base case] If $\depth(c_1)+\depth(c_2)=0$, it follows that $w_2=0$; hence $v_m;g\sat c_2$ is false.

\item[Induction step.] Assume $v_m;g\sat c_2$ with responsible branch $p^2_k = (a^2_k,b^2_k)$ (for some $k\in [1,w_2]$) and witness $h$, meaning $v_m;g=a^2_k;h$ and \tagone $h\not\sat b^2_k$. Because $m$ is a  reflective condition morphism, there is exactly one pair $(j,k) \in \Nrel_m$ for some $j\in[1,w_1]$ and a single $m'\in N_{jk} \in N$ such that $a^1_j,v_m,v_{m'},a^2_k$ is a reflection square, meaning \tagtwo $g=a^1_j;v_{m'};h$.

\smallskip
Recall that $m'\of b^2_k\to b^1_j$ is a reflective condition morphism and observe that $\depth(b^2_k)+\depth(b^1_j) < \depth(c_1)+\depth(c_2)$; hence by the induction hypothesis, $m'$ fulfills the property. Therefore, if $v_{m'};h\sat b^1_j$ then $h\sat b^2_k$, which contradicts \tagone; it follows that $v_{m'};h\not\sat b^1_j$,  which with \tagtwo allows us to conclude $g\sat p^1_j$, and hence $g\sat c_1$. 
\end{description}
\end{proofE}

\noindent
The following example shows a case in which a reflection morphism can indeed serve as a proof of entailment between conditions, but also a case in which there is entailment but no reflection morphism.\ifarxiv
\footnote{In the published version \cite{GCM-original}, this example contains an error: it is there stated that the morphism goes from $c_2$ to $c_3$, and indeed the dashed arrows are (incorrectly) drawn in the reverse direction.}
\fi

\begin{example}\label{ex:reflective morphism}
In \cref{ex:satisfaction} we have argued that there exist entailments $c_1\entails c_2$ and $c_2\entails c_3$ for the conditions of \cref{fig:conditions}. In fact, there exists a reflective morphism 
\ifarxiv
$m\of c_3\to c_2$\ignorespaces
\else
$m\of c_2\to c_3$\ignorespaces
\fi
, depicted by the dashed arrows in the following diagram, which by \cref{prop:reflective}, given that $v_m=\id$, establishes $c_2\entails c_3$:
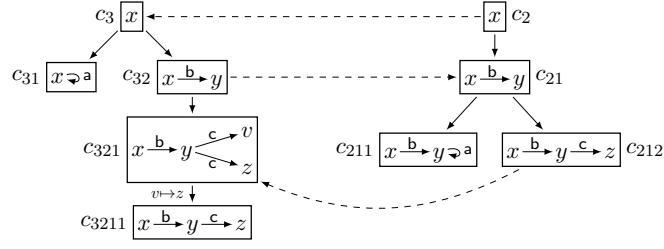
\[\vspace*{-1mm}\tikzscale{\begin{tikzpicture}[on grid]
  \node[graph] (30) {\onenode{x}};
  \node[left=0 of 30.west,inner sep=1] {$c_3$};
  \node[graph,below left=of 30] (31) {\oneloop{x}{a}}; 
  \node[left=0 of 31.west,inner sep=1] {$c_{31}$};
  \node[graph,below right=of 30] (32) {\oneedge{x}{b}{y}}; 
  \node[left=0 of 32.west,inner sep=1] {$c_{32}$};
  \node[graph,below=1.2 of 32] (321) {\onetwoedge{x}{b}{y}{c}{v}{c}{z}}; 
  \node[left=0 of 321.west,inner sep=1] {$c_{321}$};
  \node[graph,below=1.2 of 321] (3211) {\twoedge{x}{b}{y}{c}{z}}; 
  \node[left=0 of 3211.west,inner sep=1] {$c_{3211}$};

  \path (30) edge[->] (31)
        (30) edge[->] (32)
		(32) edge[->] (321)
		(321) edge[->] node[left] {\mapping{v&z}} (3211);

  \node[graph] (20) [right=6 of 30] {\onenode{x}};
  \node[right=0 of 20.east,inner sep=1] {$c_2$};
  \node[graph,below=of 20] (21) {\oneedge{x}{b}{y}}; 
  \node[right=0 of 21.east,inner sep=1] {$c_{21}$};
  \node[graph,below left=1.2 and 1.1 of 21] (211) {\oneedgeloop{x}{b}{y}{a}};
  \node[left=0 of 211.west,inner sep=1] {$c_{211}$};
  \node[graph,below right=1.2 and 1.1 of 21] (212) {\twoedge{x}{b}{y}{c}{z}};
  \node[right=0 of 212.east,inner sep=1] {$c_{212}$};
  
  \path (20) edge[->] (21)
        (21) edge[->] (211)
        (21) edge[->] (212);
\ifarxiv
\path[dashed]
  (20) edge[->] (30)
  (32) edge[->] (21)
  (212) edge[->,bend left=25] (321);
\else
\path[dashed]
  (30) edge[->] (20)
  (21) edge[->] (32)
  (321) edge[->,bend right=25] node[below,inner sep=0] {$\mapping{v&z}$} (212);
\fi
\end{tikzpicture}}\vspace*{-2mm}\]
In this case, the two subdiagrams that need to be reflection squares both commute and have identities as their top arrow; it follows from \cref{prop:reflection square}.\ref{iso simple reflection square} below that they are indeed reflection squares (\emph{simple} ones, in terms of that proposition).

There does \emph{not}, however, exist a morphism from $c_2$ to $c_1$ that establishes the entailment $c_1\entails c_2$, even though that entailment holds. To see this, note that such a morphism would have to include arrows from both $c_{111}=\inline{\looponeedge b x a y}$ and $c_{112}=\inline{\looponeedge b x c y}$ to either $c_{211}=\inline{\oneedgeloop x b y a}$ or $c_{212}=\inline{\twoedge x b y c z}$ (see \cref{fig:conditions}), which obviously do not exist.\qed
\end{example}

\begin{propositionE}[preservative morphisms preserve satisfaction]\label{prop:preservative}
Let $c_i\in \AC{R_i}$ be  conditions for $i=1,2$. If $m\of c_1\to c_2$ is a preservative condition  morphism, then $g\sat c_1$ implies $v_m;g \sat c_2$ for all arrows $g\of R_1\to G$. 
\end{propositionE}

\begin{proofE}
By induction on $\depth(c_1)+\depth(c_2)$. Let $c_i=(R_i,p^i_1\cdots p^i_{w_i})$ such that $p^i_j=(a^i_j,b^i_j)$ for  $i=1,2$ and all $j\in[1,w_i]$, and let $m=(v_m,\Nrel_m,N_m)$.
\begin{description}
\item[Base case] If $\depth(c_1)+\depth(c_2)=0$, it follows that $|c_1|=0$; hence $g\sat c_1$ is false. 
\item[Induction step.] Assume $g\sat c_1$, due to $g\sat p^1_j$ for some $j\in [1,w_1]$; and let $h$ be the witness, meaning $g=a^1_j;h$ and \tagone $h\not\sat b^1_j$. Because $m$ is a preservative condition morphism, there is exactly one pair $(j,k) \in \Nrel_m$ for some $k\in[1,w_2]$ and a single $m'\in N^{jk} \in N$ such that $a^1_j,v_m,v_{m'},a^2_k$ form a preservation square, meaning that there is an $h'$ such that $h=v_{m'};h'$ and  \tagtwo  $v_m;g=a^2_k;h'$. 

\smallskip
Recall that $m'\of b^2_k\to b^1_j$ is a preservative condition morphism and observe that $\depth(b^2_k)+\depth(b^1_j) < \depth(c_1)+\depth(c_2)$; hence by the induction hypothesis, $m'$ fulfils the property. Therefore, if $h'\sat b^2_k$ then $h=v_{m'};h'\sat b^1_j$, which contradicts \tagone; it follows that $h'\not\sat b^2_k$, from which together with \tagtwo we may conclude $v_m;g\sat p^2_k$ and hence $v_m;g\sat c_2$. 
\end{description}
\end{proofE}
\noindent
Thus, both reflective and preservative morphisms provide proof techniques for entailment: in particular, every root-preserving reflective morphism $m\of c_1\to c_2$ as well as  every root-preserving preservative morphism $m\of c_2\to c_1$ constitutes evidence for $c_2\entails c_1$. However, it turns out that reflective morphisms are strictly more powerful than preservative morphisms, in the following sense.

\begin{restatable}[Preservative morphism inversion]{proposition}{preservativeMorphismInversion}%
\label{prop:reflective beats preservative}
If $m\of c\to c'$ is a preservative morphism of which $v_m$ is split mono, then there is a reflective morphism $m'\of c'\to c$ such that, moreover, $m;m'=\id_c$.
\end{restatable}

\begin{textAtEnd}
We postpone the proof of this result until we have introduced some further properties of morphisms in the next section, which are needed for the proof.
\end{textAtEnd}

\subsection{Other properties of squares and morphisms}
\label{sec:squares}

We present here some useful characterizations of reflection and preservation squares, as well as some properties that are used, for example, to show that reflective and preservative morphisms compose. 

\begin{propositionE}[reflection squares]\label{prop:reflection square}~
\begin{enumerate}[nosep]
\item\label{reflection-pushout} The left hand diagram of \eqref{eq:reflection square prop} is a reflection square if and only if, for the pushout $g,h$ of the span $v,a_2$ as drawn in the middle diagram, the bottom left subdiagram commutes, i.e., $g=a_1;v';h$.

\begin{tikzequation}\label{eq:reflection square prop}
\tikzscale{\begin{tikzpicture}[on grid,node distance = 1.5cm,baseline=(a1)]
\node (R1) {$R_1$};
\node (P1) [below=of R1] {$P_1$};
\node (R2) [right=of R1] {$R_2$};
\node (P2) [below=of R2] {$P_2$};

\path
  (R1) edge [->] node (a1) [auto] {$a_1$} (P1)
  (R2) edge [->] node [auto] {$a_2$} (P2)
       edge [->] node [above] {$v$} (R1)
  (P1) edge [->] node [auto] {$v'$} (P2);
\end{tikzpicture}}
\qquad
\tikzscale{\begin{tikzpicture}[on grid,node distance = 1.5cm,baseline=(a1)]
\node (R1) {$R_1$};
\node (P1) [below=of R1] {$P_1$};
\node (R2) [right=of R1] {$R_2$};
\node (P2) [below=of R2] {$P_2$};

\path
  (R1) edge [->] node (a1) [auto] {$a_1$} (P1)
  (R2) edge [->] node [auto] {$a_2$} (P2)
       edge [->] node [above] {$v$} (R1)
  (P1) edge [->] node [auto] {$v'$} (P2);

\node (G) [below left=of P1] {$G$};
\path
  (R1) edge [->,bend right=25] node[left] {$g$} (G)
  (P2) edge [->,bend left=25] node[pos=.3,auto,inner sep=2] {$h$} (G);
   
\path (G) edge[pushout] +(+3.5mm,3.5mm);
\path (G) edge[white] node[pos=.7,black] {$\circlearrowright$} (P1);
\end{tikzpicture}}
\qquad
\tikzscale{\begin{tikzpicture}[on grid,auto,node distance = 1.5cm,baseline=(a1)]
\node (R1) {$R_1$};
\node (P1) [below=of R1] {$P_1$};
\node (R2) [right=of R1] {$R_2$};
\node (P2) [below=of R2] {$P_2$};
\path
  (R1) edge [->] node (a1) [auto] {$a_1$} (P1)
  (R2) edge [->] node [auto] {$a_2$} (P2)
       edge [->] node [above] {$v$} (R1)
  (P1) edge [->] node [auto] {$v'$} (P2);
\node (E) [right=of R2] {$E$};
\path (E) edge [->] node[above] {$e$} (R2);
\end{tikzpicture}}
\end{tikzequation}%
\end{enumerate}
The left hand diagram of \eqref{eq:reflection square prop} is called a \emph{simple reflection square} if there is an equalising arrow $e$ as in the right hand diagram (i.e., such that $e;a_2=e;v;a_1v'$) such that $e;v$ is epi.
\begin{enumerate}[resume,nosep]
\item Every simple reflection square is a reflection square.

\item\label{iso simple reflection square} If the left hand diagram of \eqref{eq:reflection square prop} commutes (i.e., $a_2=v;a_1;v'$) and $v$ is epi, then it is a simple reflection square.

\item\label{simple reflection square equaliser} In a simple reflection square, if $e$ is the equaliser of $a_2$ and $v;a_1;v'$ then $e;v$ is epi.

\item In the outer diagram of \eqref{eq:reflection square section}, if $v$ is split epi with section $s$ such that $s;a_2=a_1;v'$, then the diagram is a simple reflection square.

\vspace*{-3mm}\begin{tikzequation}\label{eq:reflection square section}
\tikzscale{\begin{tikzpicture}[on grid,node distance = 1.5cm,baseline=(a1)]
\node (R1) {$R_1$};
\node (P1) [below=2 of R1] {$P_1$};
\node (R2) [right=3 of R1] {$R_2$};
\node (P2) [below=2 of R2] {$P_2$};

\path
  (R1) edge [->] node [left] (a1) {$a_1$} (P1)
  (R2) edge [->] node [auto] {$a_2$} (P2)
       edge [->] node [above] (v) {$v$} (R1)
  (P1) edge [->] node [below] (v') {$v'$} (P2);
  
\node (S) [below right=1 and 1.5 of R1] {$R_1$};

\path (S) edge[->] node[below left,pos=.4,inner sep=3] {$\id$} (R1)
          edge[->] node[below right,pos=.4,inner sep=3] {$s$} (R2);
          
\path (S) edge[white] node[black] {$\circlearrowright$} (v')
          edge[white] node[black] {$\circlearrowright$} (v);
\end{tikzpicture}}
\end{tikzequation}%

\item Simple reflection squares compose; i.e., in
\begin{tikzequation}\label{eq:reflection squares compose}
\tikzscale{\begin{tikzpicture}[on grid,node distance = 1.5cm,baseline=(a1)]
\node (R1) {$R_1$};
\node (P1) [below=of R1] {$P_1$};
\node (R2) [right=of R1] {$R_2$};
\node (P2) [below=of R2] {$P_2$};
\node (R3) [right=of R2] {$R_3$};
\node (P3) [below=of R3] {$P_3$};

\path
  (R1) edge [->] node (a1) [auto] {$a_1$} (P1)
  (R2) edge [->] node [auto] {$a_2$} (P2)
       edge [->] node [above] {$v_1$} (R1)
  (P1) edge [->] node [auto] {$v_1'$} (P2)
  (R3) edge [->] node [auto] {$a_3$} (P3)
       edge [->] node [above] {$v_2$} (R2)
  (P2) edge [->] node [auto] {$v_2'$} (P3);
\end{tikzpicture}}
\end{tikzequation}%
if the left and right hand subdiagrams are simple reflection squares, then so is the outer subdiagram.

\item In a reflection square, if the upper arrow $v$ is an isomorphism, then the diagram commutes.

\item Reflection squares compose; i.e., in \eqref{eq:reflection squares compose}, 
if both the left hand and the right hand subdiagrams are reflection squares, then so is the outer subdiagram. \label{reflection squares compose}

\item In a reflection square as in \eqref{eq:reflection square}, the arrow $v$ is epi. (Proof due to Daniel Schepler \cite{mathexchange}.)
\end{enumerate}
\end{propositionE}

\begin{proofE}~
\begin{enumerate}[defsep]
\item (Only if) Obvious, as the pushout arrows $g,h$ make the outer diagram commute. (If) For a pair $g', h'$ as in the middle diagram of \eqref{eq:reflection square}, there exists a (unique) $k\of G\to G'$ (where $G$ is the pushout object as in \eqref{eq:reflection square prop}) such that $g'=g;k$ and $h'=h;k$; hence $g'= g;k= a_1; v';h;k=a_1;v';h'$.

  \item The equalising arrow guarantees $e;a_2;h'=e;v;a_1;v';h'$, and the commuting outer subdiagram guarantees $e;a_2;h'=e;v;g'$; hence $e;v;a_1;v';h'=e;v;g'$. Since $e;v$ is epi, it follows that $a_1;v';h'=g'$.

\item Since the diagram already commutes, $e=\id$ is an equaliser and $e;v=v$ is epi, as required.

\item In the right hand diagram of \eqref{eq:reflection square prop}, let $e$ be such that $e;a_2=e;v;a_1;v'$ and $e;v$ is epi, and consider the equaliser $e'\of E'\to R_2$ of $a_2$ and $v;a_1;v'$. By the universal property of $e'$, there is a (unique) $k\of E\to E'$ such that $e=k;e'$. But then, since $e;v$ ($=k;e';v$) is epi, so is $e';v$.

\[\tikzscale{\begin{tikzpicture}[on grid,auto,node distance = 1.5cm,baseline=(a1)]
\node (R1) {$R_1$};
\node (P1) [below=of R1] {$P_1$};
\node (R2) [right=of R1] {$R_2$};
\node (P2) [below=of R2] {$P_2$};

\path
  (R1) edge [->] node (a1) [left] {$a_1$} (P1)
  (R2) edge [->] node [left] {$a_2$} (P2)
       edge [->] node [above] {$v$} (R1)
  (P1) edge [->] node [auto] {$v'$} (P2);

\node (E) [right=of R2] {$E$};
\path (E) edge [->] node[above] {$e$} (R2);

\node (E') [below=of E] {$E'$};
\path (E') edge [right hook->] node[pos=.4,inner sep=1] {$e'$} (R2)
      (E) edge [dashed,->] node[right] {$k$} (E');
\end{tikzpicture}}\]

\item The section $s$ in \eqref{eq:reflection square section} satisfies the conditions for $e$ in the right hand diagram of \eqref{eq:reflection square prop} (noting that $s;v=\id$ is epi).

\item Consider the following diagram, where $e_1$ is such that \tagone $e_1;a_2=e_1;v_1;a_1;v_1'$ and $e_1;v_1$ is epi, $e_2$ is such that \tagtwo $e_2;a_3=e_2;v_2;a_2;v_2'$ and $e_2;v_2$ is epi, and $e_1',e_2'$ form the pullback of $e_1$ and $e_2;v_2$.
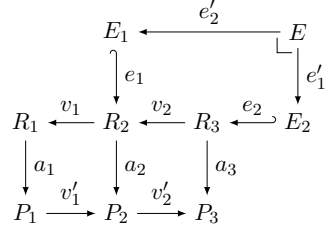
\[\tikzscale{\begin{tikzpicture}[on grid,auto,node distance = 1.5cm,baseline=(a1)]
\node (R1) {$R_1$};
\node (P1) [below=of R1] {$P_1$};
\node (R2) [right=of R1] {$R_2$};
\node (P2) [below=of R2] {$P_2$};
\node (R3) [right=of R2] {$R_3$};
\node (P3) [below=of R3] {$P_3$};

\path
  (R1) edge [->] node (a1) [auto] {$a_1$} (P1)
  (R2) edge [->] node [auto] {$a_2$} (P2)
       edge [->] node [above] {$v_1$} (R1)
  (P1) edge [->] node [auto] {$v_1'$} (P2)
  (R3) edge [->] node [auto] {$a_3$} (P3)
       edge [->] node [above] {$v_2$} (R2)
  (P2) edge [->] node [auto] {$v_2'$} (P3);
  
\node (E1) [above=of R2] {$E_1$};
\node (E2) [right=of R3] {$E_2$};
\node (E) [above=of E2] {$E$};

\path
  (E) edge [->] node {$e_1'$} (E2)
      edge [->] node[above] {$e_2'$} (E1)
  (E1) edge [left hook->] node {$e_1$} (R2)
  (E2) edge [left hook->] node[above] {$e_2$} (R3);

\path (E) edge[pushout] +(-3.5mm,-3.5mm);
\end{tikzpicture}}\]
Using also that \tagthree the pullback square commutes, we can deduce
\begin{align*}
e'_1;e_2;v_2;v_1;a_1;v_1';v_2'
 & \stackrel{\tagthree}{=} e_2';e_1;v_1;a_1;v_1';v_2' \\
 & \stackrel{\tagone}= e_2';e_1;a_2;v_2' \\
 & \stackrel{\tagthree}= e_1';e_2;v_2;a_2;v_2' \\
 & \stackrel{\tagtwo}= e_1';e_2;a_3
\end{align*}
Moreover, epis are preserved by pullback (in presheaf toposes), hence $e_2'$ is epi. Finally, epis compose, hence $e_1';e_2;v_2;v_1=e_2';e_1;v_1$ is epi. Thus, $e_1';e_2$ satisfies the conditions for $e$ in the definition of simple reflection squares.

\item Consider the right diagram of \eqref{eq:reflection square} in case $v$ is an isomorphism. 
Let \tagone $h'=id_{P_2}$ and \tagtwo $g'= v^{-1};a_2$: they make the outer diagram commute, thus \tagthree $g' = a_1; v'; h'$ because the inner one is a reflection square. It follows that $v;a_1;v' \stackrel{\tagone}=  v;a_1;v';h' \stackrel{\tagthree}= v;g' \stackrel{\tagtwo}= v;v^{-1};a_2 = a_2$, hence the diagram commutes.

\item Consider any pair of arrows $g',h'$ making the outer diagram commute:

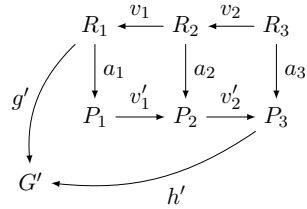
\[\tikzscale{\begin{tikzpicture}[auto,on grid,node distance = 1.5cm,baseline=(a1)]
\node (R1) {$R_1$};
\node (P1) [below=of R1] {$P_1$};
\node (R2) [right=of R1] {$R_2$};
\node (P2) [below=of R2] {$P_2$};
\node (R3) [right=of R2] {$R_3$};
\node (P3) [below=of R3] {$P_3$};

\path
  (R1) edge [->] node (a1) [auto] {$a_1$} (P1)
  (R2) edge [->] node [auto] {$a_2$} (P2)
       edge [->] node [above] {$v_1$} (R1)
  (P1) edge [->] node [auto] {$v_1'$} (P2)
  (R3) edge [->] node [auto] {$a_3$} (P3)
       edge [->] node [above] {$v_2$} (R2)
  (P2) edge [->] node [auto] {$v_2'$} (P3);

\node (G) [below left=of P1] {$G'$};
\path
  (R1) edge [->,bend right=25] node[left] {$g'$} (G)
  (P3) edge [->,bend left=20] node[auto] {$h'$} (G);
  
\end{tikzpicture}}\]
Since the right one is a reflection square we have $v_1;g'=a_2;v_2';h'$; then since the left one is also a reflection square we have $g'=a_1;v'_1;v_2';h'$, as required.

\item The proof can be found \href{https://math.stackexchange.com/questions/5129159/does-the-existence-of-a-commuting-diagonal-in-a-pushout-imply-that-one-leg-is-ep}{here}.
\end{enumerate}
\end{proofE}
%
%
\begin{example}\label{ex:reflection square}
To see that the first diagram shown in \cref{ex:square} (left hand side of \eqref{eq:ex-squares}) is a reflection square, note that in the following (left hand) diagram, where $g,h$ form a pushout for $v,a_2$, we have $g=a_1;v';h$ and \cref{prop:reflection square}.\ref{reflection-pushout} applies. It is not a \emph{simple} reflection square because the equaliser $e$ of $v;a_1;v'$ and $a_2$ maps a single-node discrete graph to node $z$ of $R_2$, hence $e;v$ is not epi.
\vspace*{-2mm}\[\tikzscale
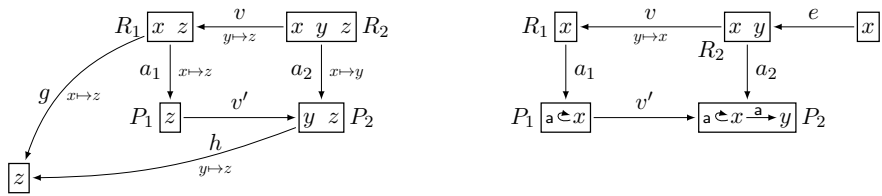
{\begin{tikzpicture}[on grid,auto,node distance = 1.5cm,baseline=(a)]
\node[graph] (R1) {\twonode x z};
\node[left=0 of R1.west,inner sep=1] {$R_1$};
\node[graph,below=of R1] (P1) {\onenode z};
\node[left=0 of P1.west,inner sep=1] {$P_1$};
\node[graph,right=2.5 of R1] (R2) {\threenode x y z}; 
\node[right=0 of R2.east,inner sep=1] {$R_2$};
\node[graph,below=of R2] (P2) {\twonode y z};
\node[right=0 of P2.east,inner sep=1] {$P_2$};
  
\path
   (R1) edge[->] node[left] (a1) {$a_1$} node[right] {$\mapping{x&z}$} (P1)
   (R2) edge[->] node[left] {$a_2$} node[right] {$\mapping{x&y}$} (P2)
        edge[->] node[above] {$v$} node[below,inner sep=0] {$\mapping{y&z}$} (R1)
   (P1) edge[->] node {$v'$} (P2);

\node[graph,below left=1 and 2.5 of P1] (PO) {$\onenode z$};

\path (R1) edge[->,bend right=25] node[left,pos=.6] {$g$} node[right,pos=.6] {$\mapping{x&z}$} (PO)
      (P2) edge[->,bend left=10,inner sep=2] node[above,pos=.3] {$h$} node[below,pos=.3] {$\mapping{y&z}$} (PO); 
\end{tikzpicture}}
\qquad\qquad
\tikzscale{\begin{tikzpicture}[on grid,auto,node distance = 1.5cm,baseline=(a)]
\node[graph] (R1) {\onenode x};
\node[left=0 of R1.west,inner sep=1] {$R_1$};
\node[graph,below=of R1] (P1) {\oneloopleft x a};
\node[left=0 of P1.west,inner sep=1] {$P_1$};
\node[graph,right=3 of R1] (R2) {\twonode x y}; 
\node[below left=0 of R2.south west,inner sep=-2] {$R_2$};
\node[graph,below=of R2] (P2) {\looponeedge x a a y};
\node[right=0 of P2.east,inner sep=1] {$P_2$};
  
\path
   (R1) edge[->] node (a) {$a_1$} (P1)
   (R2) edge[->] node {$a_2$} (P2)
        edge[->] node[above] {$v$} node[below,inner sep=0] {$\mapping{y&x}$} (R1)
   (P1) edge[->] node {$v'$} (P2);

\node[graph,right=2 of R2] (E) {\onenode x};

\path (E) edge[->] node[above] {$e$} (R2); 
\end{tikzpicture}}\]
The second diagram of \cref{ex:square} (centre of \eqref{eq:ex-squares}), repeated on the right hand above, is a \emph{simple} reflection square: the equalizer $e$ of $v;a_1;v'$ and $a_2$ is such that $e;v$ is epi.\qed 
\end{example}

\noindent
The following proposition states some useful facts regarding preservation squares.
\begin{propositionE}[preservation squares]\label{prop:preservation square}~
\begin{enumerate}[nosep]
\item\label{preservation square any} In the left hand diagram of \eqref{eq:preservation square}, the diagram is a preservation square if and only if $v'$ is a split mono with a retract $r'$ (meaning $v';r'=\id_{P_1}$) such that the upper pentagon in \eqref{eq:preservation square any} commutes (meaning $a_2;r'=v;a_1$).

\vspace*{-5mm}\begin{tikzequation}\label{eq:preservation square any}
\tikzscale{\begin{tikzpicture}[on grid,node distance = 1.5cm,baseline=(a1)]
\node (R1) {$R_1$};
\node (P1) [below=2 of R1] {$P_1$};
\node (R2) [right=3 of R1] {$R_2$};
\node (P2) [below=2 of R2] {$P_2$};

\path
  (R1) edge [->] node [left] (a1) {$a_1$} (P1)
  (R2) edge [->] node [auto] {$a_2$} (P2)
       edge [->] node [above] (v) {$v$} (R1)
  (P1) edge [->] node [below] (v') {$v'$} (P2);
  
\node (S) [above right=1 and 1.5 of P1] {$P_1$};

\path (S) edge[<-] node[above left,pos=.4,inner sep=3] {$\id$} (P1)
          edge[<-] node[above right,pos=.4,inner sep=3] {$r'$} (P2);
          
\path (S) edge[white] node[black,pos=.4] {$\circlearrowright$} (v')
          edge[white] node[black] {$\circlearrowright$} (v);
\end{tikzpicture}}
\end{tikzequation}%

\item In the left hand diagram of \eqref{eq:preservation square}, if the bottom arrow $v'$ is iso then the diagram is a preservation square if and only if it commutes.

\item\label{preservation square compose} Preservation squares compose; i.e., in
\vspace*{-2mm}\begin{tikzequation}\label{eq:preservation squares compose}
\tikzscale{\begin{tikzpicture}[on grid,node distance = 1.5cm,baseline=(a1)]
\node (R1) {$R_1$};
\node (P1) [below=of R1] {$P_1$};
\node (R2) [right=of R1] {$R_2$};
\node (P2) [below=of R2] {$P_2$};
\node (R3) [right=of R2] {$R_3$};
\node (P3) [below=of R3] {$P_3$};

\path
  (R1) edge [->] node (a1) [auto] {$a_1$} (P1)
  (R2) edge [->] node [auto] {$a_2$} (P2)
       edge [->] node [above] {$v_1$} (R1)
  (P1) edge [->] node [auto] {$v_1'$} (P2)
  (R3) edge [->] node [auto] {$a_3$} (P3)
       edge [->] node [above] {$v_2$} (R2)
  (P2) edge [->] node [auto] {$v_2'$} (P3);
\end{tikzpicture}}
\end{tikzequation}%
if the left and right hand subdiagrams are preservation squares, then so is the outer subdiagram.
\end{enumerate}
\end{propositionE}

\begin{proofE}~
\begin{enumerate}[defsep]
\item For the \emph{only if} part, if the outer square of \eqref{eq:preservation square any} is a preservation square, then note that $g=a_1$ and $h=\id_{P_1}$ are a particular pair of arrows such that $g=a_1;h$. Therefore there is an $h'$ such that $h = \id_{P_1}=v';h'$, which means that $v'$ is split mono and $h'$ is its retract, and $v;g=a_2;h'$, hence $v;a_1 = v;g = a_2; h'$, as required. 
For the \emph{if} part, assume that $r'$ is the retract of $v'$ and $a_2;r' = v; a_1$. Then for any pair $g,h$ such that $g = a_1;h$ define $h'=r';h$; then $v';h'=v';r';h=h$ and $v;g=v;a_1;h=a_2;r';h=a_2;h'$ as required.

\item  If $v'$ is iso in the left hand diagram of \eqref{eq:preservation square}, then it is split mono with retract $r'=v'^{-1}$. If the diagram commutes, namely $a_2=v;a_1;v'$, then $a_2;r'=v_1;a_1;v';r' = v_1;a_1$, thus it is a preservation square by the previous point. Vice versa, $a_2;r'=v_1;a_1$ implies $a_2=a_2;r';v'=v_1;a_2;v'$, proving commutativity of the original square.

\item We exploit the characterization of point 1. If both subdiagrams are preservation squares  and the split monos $v_i'$ have retracts $r_i'$ for $i=1,2$, then $v_1';v_2'$ is also split mono, with retract $r_2';r_1'$. By applying $v_i;a_i=a_{i+1};r_i'$ for $i=1,2$, it follows that $v_2;v_1;a_1=v_2;a_2;r_1'=a_3;r_2';r_1'$.
\end{enumerate}
\end{proofE}
\begin{example}\label{ex:preservation square}
To see that the diagram shown in \cref{ex:square} (right hand side of \eqref{eq:ex-squares}, repeated below) is a preservation square, note that $r'$ is a retract of $v'$ and $v;a_1=a_2;r'$; hence \cref{prop:preservation square}.\ref{preservation square any} applies.

\vspace*{-3mm}\[\tikzscale{\begin{tikzpicture}[on grid,auto,node distance = 1.5cm]
\node[graph] (R1) {\nodeedge z x a y};
\node[left=0 of R1.west,inner sep=1] {$R_1$};
\node[graph,below=2 of R1] (P1) {\oneedge x a y};
\node[left=0 of P1.west,inner sep=1] {$P_1$};
\node[graph,right=5 of R1] (R2) {\oneedge x a y}; 
\node[right=0 of R2.east,inner sep=1] {$R_2$};
\node[graph,below=2 of R2] (P2) {\spangraph z a x a y};
\node[right=0 of P2.east,inner sep=1] {$P_2$};
  
\path
   (R1) edge[->] node[left] {$a_1$} node[right] {$\mapping{z&x}$} (P1)
   (R2) edge[->] node[right] {$a_2$} node[left] {$\mapping{y&z}$} (P2)
        edge[->] node[above] {$v$} (R1)
   (P1) edge[->] node[below] {$v'$} (P2);

\node[graph,above right=1 and 2.5 of P1] (R) {\oneedge x a y};

\path
   (P1) edge[->] node[above] {$\id$} (R)
   (P2) edge[->] node[above] {$r'$} node[pos=.8,below] {$\mapping{z&y}$} (R);
\end{tikzpicture}}\]\raisedqed[14]
\end{example}

\noindent
We exploit the results just introduced to show that reflective and preservative morphisms compose.

\begin{propositionE}\label{prop:reflective-preservative composition}
Let $c\in\AC R$ and $c_i\in\AC{R_i}$ for $i=1,2,3$ and let $m_i\of c_i\to c_{i+1}$ be condition morphisms for $i=1,2$.
\begin{enumerate}[nosep]
\item $\id_c$ is a reflective condition  morphism, and if $m_1$ and $m_2$ are reflective then so is $m_1;m_2$.
\item $\id_c$ is a preservative condition  morphism, and if $m_1$ and $m_2$ are preservative then so is $m_1;m_2$.
\end{enumerate}
\end{propositionE}
\begin{proofE}~
\begin{enumerate}[defsep]
\item Morphism $\id_c$ (see \eqref{eq:id condition morphism})
obviously fulfills the requirement of reflectiveness, noting in particular that $a,\id,\id,a$ is a reflection square. Now consider morphism $m_1;m_2$ as defined in \eqref{eq:condition morphism composition}, and let us proceed by induction on $\depth(c_1) + \depth(c_2) + \depth(c_3)$. Since by  assumption relation $I_1$ is a function from $[1,w_2]$ to $[1,w_1]$ and $I_2$ is a function from $[1,w_3]$ to $[1,w_2]$, then the composed relation  $I_1;I_2$ is a function from $[1,w_3]$ to $[1,w_1]$. This implies that for each $(i,k) \in I_1;I_2$ there is a unique $j\in [1,w_2]$ such that $(j,k) \in I_2$ (and $(i,j) \in I_1$), and this, together with $|N^1_{ij}| = |N^2_{jk}| = 1$ which holds by assumption, ensures $|N_{ik}| = 1$. The only morphism of $N_{ik}$ is the composition of two morphisms which are reflective by assumption and relating three conditions of total depth smaller than $\depth(c_1) + \depth(c_2) + \depth(c_3)$,  thus it is reflective by induction hypothesis. Finally, the induced square is the composition of two reflection squares and thus is a reflection square by  \cref{prop:reflection square} \ref{reflection squares compose}.

\item The proof is analogous to that of point 1. 
\end{enumerate}
\end{proofE}
\begin{textAtEnd}
We conclude this section with the proof of the earlier \cref{prop:reflective beats preservative} showing that, under mild assumptions, for each preservative morphism between two conditions there is a reflective morphism in the opposite direction that is a retract of it.

\preservativeMorphismInversion*

\begin{proof}[Proof of \cref{prop:reflective beats preservative}]
By induction on $\depth(c)+\depth(c')$.
\begin{description}[defsep]
\item[Base case.] If $\depth(c)+\depth(c')=0$, then $m=(v,\varnothing,\varnothing)$ where $v: R^{c'} \to R^c$  is split mono; hence there is a retract $r\of R^{c'}\to R^c$ such that $v;r=\id_{R^c}$. Let $m'=(r,\varnothing,\varnothing)$. Clearly $m'$ is a reflective morphism from $c'$ to $c$ and $m;m'=\id_c$.
 
\item[Induction step.]  Assume the property holds up to $\depth(c)+\depth(c')$, and let $c=(R,(a_1,b_1)\ccdots (a_w,b_{w}))$, $c'=(R',(a'_1,b'_1)\ccdots (a'_{w'}, b'_{w'}))$ and  $m=(v,\Nrel,N)$. Since $m$ is preservative,  $\Nrel$ is a function from $[1,w]$ to $[1,w']$ and for each $(i,j) \in \Nrel$, $N_{ij} \isdef \setof{m_i: b'_j \to b_i}$ is a singleton.

\smallskip
For all $i\in[1,w]$, because $v$ is a split mono with retract $r$ and $m_i: b'_j \to b_i$ is preservative, the diagram on the left exists, where the top triangle commutes and $v;a_i=a'_j;r_{i}$:
\[\tikzscale{\begin{tikzpicture}[on grid,node distance = 1.5cm,baseline=(a1)]
\node (R1) {$R^c$};
\node (P1) [below=2 of R1] {$P_i$};
\node (R2) [right=3 of R1] {$R^{c'}$};
\node (P2) [below=2 of R2] {$P'_j$};

\path
  (R1) edge [->] node (a1) [left] {$a_i$} (P1)
  (R2) edge [->] node [auto] {$a'_j$} (P2)
       edge [->] node [below] {$v$} (R1)
  (P1) edge [->] node [below] (v') {$v_{m_{i}}$} (P2);

\node (S') [above right=1 and 1.5 of P1] {$P^p_j$};  

\path (S') edge[<-] node[above left,pos=.4,inner sep=2] {$\id$} (P1)
          edge[<-] node[above right,pos=.4,inner sep=2] {$r_{i}$} (P2);

\node (S) [above right=1 and 1.5 of R1] {$R^{c'}$};  

\path (S) edge[<-] node[above left,pos=.4,inner sep=3] {$r$} (R1)
          edge[<-] node[above right,pos=.4,inner sep=3] {$\id$} (R2);
\end{tikzpicture}}
\qquad\qquad
\tikzscale{\begin{tikzpicture}[on grid,auto,node distance = 1.5cm,baseline=(a1)]
\node (R1) {$R^{c'}$};
\node (P1) [below=2 of R1] {$P'_j$};
\node (R2) [left=2 of R1] {$R^{c'}$};
\node (P2) [below=2 of R2] {$P^i$};

\path
  (R1) edge [->] node [right] {$a_i$} (P1)
  (R2) edge [->] node [left] {$a'_j$} (P2)
  (R1) edge [->] node {$r$} (R2)
  (P2) edge [->] node [below] {$r_{i}$} (P1);

\node (S) [right=2 of R1] {$R^c$};
\path (S) edge [->] node {$v$} (R1);
\end{tikzpicture}}
\]
Therefore we have $v;a_i=a'_j;r_{i}=v;r;a'_j;r_{i}$, meaning that $v$ equalizes the right hand diagram above. Since, moreover, $v;r=\id$ is epi, the right hand diagram is a (simple) reflection square.

\smallskip
Note that $m_{i}\of b'_j\to b_i$ is a preservative morphism with top arrow $v_{m_i}$ a split mono. Since $\depth(b'_j)+\depth(b_i)<\depth(c)+\depth(c')$, by the induction hypothesis there is a reflective morphism $m'_{i}\of b_i\to b'_j$ such that $m_i;m'_i=\id_{b'_j}$. 
 Hence $m'=(r,\Nrel^{op},N')$ with $N'=\setof{N'_{ji} = \setof{m'_i}\mid (j,i) \in \Nrel^{op}}$ is a reflective condition morphism from $c'$ to $c$; and moreover, $m;m'=\id_c$.

\end{description}
\end{proof}
\end{textAtEnd}

\section{Categorical properties}
\label{sec:categories}

In this section, we frame the concepts and results introduced in the previous ones in a categorical setting, presenting several categories of  conditions and functors relating them. This allows us to show the functoriality of the operations on conditions introduced in \cref{sec:conditions}, with a few exceptions, and to characterize some of them as universal constructions. 

As mentioned earlier, in the literature, morphisms among  conditions are usually based on the semantic concept of entailment. Entailment (see \cref{def:entailment}) establishes a preorder on conditions over the same root, and it can be generalized to entailment of conditions over different roots related by an arrow in two different ways.

\begin{definition}[categories based on entailment]
    \label{def:categories entailment}
Given conditions $c_i \in \AC{R_i}$ for $i = 1,2$ and an arrow $f\of R_1\to R_2$ of $\cat{C}$, define $\leq^f$ as 
\[ c_1\leq^f c_2 \text{ iff for all } g\of R_2\to G: f;g\sat c_1 \text{ implies } g\sat c_2 \]
Category $\NCentailsup$ has conditions (over arbitrary roots) as objects. An arrow from $c_1$ to $c_2$ is a triple $(c_1, f\of R_{c_1} \to R_{c_2}, c_2)$ such that $c_1 \leq^f c_2$. Identities and 
composition are defined in the expected way.

\smallskip
Dually, for $c_i \in \AC{R_i}$ for $i = 1,2$ and an arrow $f\of R_2\to R_1$, define $\leq_f$ as
\[ c_1\leq_f c_2 \text{ iff for all } g\of R_1\to G: g\sat c_1 \text{ implies } f;g\sat c_2 \]
Category  $\NCentailsdown$ has conditions as objects; an arrow from $c_1$ to $c_2$ is a triple $(c_1, f\of R_{c_2} \to R_{c_1}, c_2)$ such that $c_1 \leq_f c_2$. 
Composition is defined as $(c_1, f, c_2); (c_2, g, c_3) = (c_1, g;f, c_3)$.\qed
\end{definition}
\noindent
The entailment preorders among conditions over the same root of \cref{def:entailment} are subcategories of both $\NCentailsup$ and $\NCentailsdown$.

\begin{fact}[entailment preorder]
Let $\NCentailsup(R)$ be the subcategory of $\NCentailsup$ containing as objects the conditions over $R$ and as arrows only triples of the form $(c_1, id_R, c_2)$.  Let $\NCentailsdown(R)$ be defined similarly as a subcategory of $\NCentailsdown$. Then $\NCentailsup(R) = \NCentailsdown(R)$, and both coincide with the preorder of semantic entailment of 
\cref{def:entailment}, i.e.~$(c_1, id_R, c_2)$ iff $c_1 \entails c_2$.\qed
\end{fact}
\noindent
We introduce now the categories of conditions based on the structural morphism proposed in \cref{sec:morphisms}. 

\begin{definition}[categories of structural morphisms]
    \label{def:categories structural}
The following categories are defined as having  conditions over arbitrary roots as objects, and 
\begin{itemize}[nosep]
\item $\NCplain$ has as arrows condition morphisms as defined in \cref{def:morphism};
\item $\NCrefl$ and $\NCpres$, subcategories of $\NCplain$,  have as arrows reflective and preservative morphisms,  as defined in \cref{def:reflective-preservative}.
\end{itemize}
Identities and composition are defined as in \cref{def:morphism composition}.

Given an object $R$, $\NCplain(R)$ is the subcategory of $\NCplain$ where objects are conditions over $R$, and arrows are morphisms $m$ of which $v_m=\id_R$. $\NCrefl(R)$ and $\NCpres(R)$ are defined similarly.\qed
\end{definition}

\noindent
The well-definedness of the three categories just introduced follows from \cref{prop:morphism composition,prop:reflective-preservative composition}.

\paragraph{Functors.}

The various categories of conditions are related by various, mostly obvious, functors, depicted in \cref{fig:functors}.

\begin{propositionE}[relating categories of conditions]
    \label{prop:functors}
The followings are identity-on-object functors:
\begin{enumerate}[nosep]
\item  $\cR\of {\NCrefl}^{op} \to \NCentailsup$, contravariant, mapping $(v,I,N)\of c_1 \to c_2$ to $(c_2,v,c_1)$;
\item $\cP\of \NCpres \to \NCentailsdown$, mapping $(v,I,N)\of c_1 \to c_2$ to $(c_1,v,c_2)$;
\item $\cR_R\of {\NCrefl(R)}^{op} \to \NCentailsup(R)$, contravariant, restriction of $\cR$ to  $\NCrefl(R)$;
\item $\cP_R \of \NCpres(R) \to \NCentailsdown(R) [= \NCentailsup(R)]$, restriction of $\cP$ to $\NCrefl(R)$.
\end{enumerate}
\end{propositionE}
\begin{proofE}
We show that the first two are well-defined, which implies that also the last two are:
\begin{enumerate}[nosep]
\item
Given arrow $(v,I,N)\of c_1 \to c_2$ of $\NCrefl$, $(c_2,v,c_1)$ is an arrow of $\NCentailsup$ because by \cref{prop:reflective} $v;g\sat c_2$ implies $g \sat c_1$ for all $g\of R_1\to G$.
\item Given arrow $(v,I,N)\of c_1 \to c_2$ of $\NCpres$, $(c_1,v,c_2)$ is an arrow of $\NCentailsdown$ because by \cref{prop:preservative}  $g\sat c_1$ implies $v;g \sat c_2$ for all $g\of R_1\to G$.
\end{enumerate}
\end{proofE}
\noindent
\Cref{prop:reflective beats preservative} shows that for each preservative morphism $m\of c_1 \to c_2$ having a split-mono root arrow, there is a reflective morphism $m': c_2 \to c_1$ such that $m;m' = id_{c_1}$. However
the construction involves the choice of a retract at each level of the morphism, thus the result is not uniquely determined. This implies that even restricting to conditions over the same objects with the identity as root arrow (which is split-mono) we don't obtain a proper contravariant functor from
$\NCpres(R)$ to $\NCrefl(R)$. There are two possible directions to model this situation formally: resorting to a suitable coherent choice of retractions, making the construction deterministic and ensuring functoriality, or using a weaker notion of functor that allows to map an arrow to a set of arrows. We leave this as a topic of future work.  





\begin{figure}[t]
\centering
\begin{tikzpicture}[on grid,node distance=1.5 and 1.5]
\newlength{\dx}
\setlength{\dx}{1.5cm}

\tikzset{
  inclusionover/.style={
   line width=1.2pt,
   double distance=0.8pt,
   double=gray!50,
   draw=white,
   -{Stealth[color=gray!50,length=2mm]},        
}}

\tikzset{
  inclusion/.style={
    draw=gray!50,
    line width=0.8pt,
     -{Stealth[length=2mm]},        
  }
}

\node (NC-entailsup) {$\NCentailsup$};
\node[right= of NC-entailsup, xshift=4\dx] (NC-refl) {$\NCrefl$};
\node [below right=1 and 2.5 of NC-entailsup](NC-entailsdown) {$\NCentailsdown$};
\node [below left=1 and 1.5 of NC-refl](NC-pres) {$\NCpres$};
\node[right=4 of NC-pres] (NC-plain) {$\NCplain$};
\node [below=2.5 of NC-entailsup](NC-entailsupR) {$\NCentailsup(R)$};
\node [below right=1.5 and 2.5 of NC-entailsupR)](NC-entailsupS) {$\NCentailsup(S)$};
\node [below=2.5 of NC-refl](NC-reflR) {$\NCrefl(R)$};
\node [below right=of NC-reflR](NC-reflS) {$\NCrefl(S)$};
\node [below=2.5 of NC-pres](NC-presR) {$\NCpres(R)$};
\node [below right=of NC-presR](NC-presS) {$\NCpres(S)$};
 

\path
  (NC-refl) edge[->] node[above] {$\scriptstyle{\cR^{op}}$} (NC-entailsup)
  (NC-pres) edge[->] node[above] {$\scriptstyle{\cP}$} (NC-entailsdown)
  (NC-reflR) edge[->] node[above] {$\scriptstyle{\cR_R^{op}}$} (NC-entailsupR)
  (NC-reflS) edge[->] node[pos=0.7, above] {$\scriptstyle{\cR_S^{op}}$}(NC-entailsupS)
  (NC-presR) edge[->] node[above,pos=.3] {$\scriptstyle{\cP_R}$} (NC-entailsupR)
  (NC-presS) edge[->] node[above] {$\scriptstyle{\cP_S}$} (NC-entailsupS)
    (NC-entailsupR) edge[ ->] node[above,xshift=4pt] (ar) {$\scriptstyle{\ushift{\_}u}$} (NC-entailsupS)
    (NC-entailsupS.west) edge[bend left=45, ->] node[below,xshift=-3pt,pos=.6] (al) {$\scriptstyle{\dshift{\_}u}$} (NC-entailsupR)
    (NC-reflR) edge[->] node[below,pos=0.3,xshift=-3pt] {$\scriptstyle{\ushift{\_}u}$} (NC-reflS)
    (NC-presR) edge[over,->] node[below,pos=0.3,xshift=-3pt] {$\scriptstyle{\ushift{\_}u}$} (NC-presS);

\path
  (NC-entailsupR) edge[inclusionover] (NC-entailsup)
  (NC-entailsupS) edge[inclusion] (NC-entailsup)
  (NC-entailsupS) edge[inclusionover] (NC-entailsdown)
  (NC-entailsupR) edge[inclusion] (NC-entailsdown)
  
   (NC-refl) edge[inclusionover] (NC-plain)
   (NC-reflR) edge[inclusionover] (NC-refl)
   (NC-reflS) edge[inclusionover] (NC-refl)
   (NC-presR) edge[inclusionover] (NC-pres)
   (NC-presS) edge[inclusionover] (NC-pres)
   (NC-pres) edge[inclusionover] (NC-plain);

\path (al) edge[draw=none] node[rotate=225] {$\vdash$} (ar);
\end{tikzpicture}
\caption{Overview of the functors of \cref{prop:functors} and \cref{prop:functoriality} (5-7). Gray arrows indicate inclusions of subcategories. The upshift and downshift functors are relative to an arrow $u\of S \to R$ of $\cat{C}$}.
\label{fig:functors}
\end{figure}
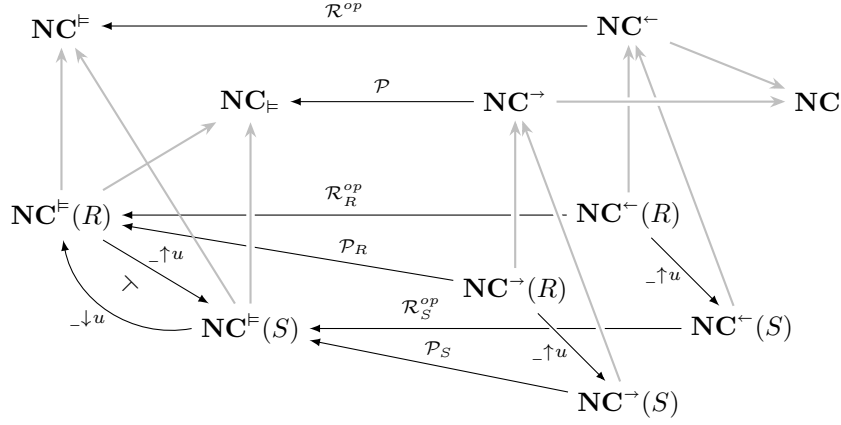

\subsection{Functoriality and universality of operations on conditions}

The categories of conditions just introduced allow us to analyse some categorical properties of the constructions on  conditions introduced in \cref{sec:conditions}. The next proposition summarizes the main universal and functoriality properties.  

\begin{propositionE}[functoriality of operations]
    \label{prop:functoriality}~
\begin{enumerate}[nosep]
\item Condition $\bot_R \isdef (R,\epsilon)$ is initial in $\NCentailsup(R)$ and in $\NCpres(R)$,  and terminal in $\NCrefl(R)$. 

\item The inverse operator $\inv{}$ induces a contravariant functor $\inv{_R} \of {\NCentailsup(R)}^{op} \to \NCentailsup(R)$.

\item The operator $\inv{}$ also induces a contravariant endofunctor on $\NCplain(R)$, which restricts to corresponding contravariant endofunctors on the subcategories $\NCrefl(R)$ and $\NCpres(R)$.

\item $\inv{_R}$ commutes with functors $\cR_R$ and $\cP_R$, i.e., $\inv{_R};\cR_R = \cR_R; \inv{_R}\of {\NCrefl(R)}^{op} \to \NCentailsup(R)$, and  $\inv{_R};\cP_R = \cP_R; \inv{_R}\of {\NCpres(R)}^{op} \to \NCentailsup(R)$.

\item For $u\of S \to R$, the upshift operator $\ushift{\_}{u}\of \NC{R}\to \NC{S}$ of \cref{def:upshift} is always functorial. More precisely, if $c_1, c_2\in \NC{R}$ and $k\of c_1 \to c_2$ is an arrow in $\NCentailsup(R), \NCplain(R), \NCrefl(R)$ or $\NCpres(R)$, then  $\ushift{k}{u}\of \ushift{c_1}{u}\to \ushift{c_2}{u}$ is an arrow in $\NCentailsup(S), \NCplain(S), \NCrefl(S)$ or $\NCpres(S)$, respectively.

\item\label{downshift-not-functorial} For $v: R \to S$, the downshift operator $\dshift{\_}{v}:  \NC{R} \to \NC{S}$ of \cref{def:downshift} is functorial on $\NCentailsup(R)$, but it is not functorial on $\NCplain(R)$, $\NCrefl(R)$ or $\NCpres(R)$.
\item For $u\of R \to S$, functor $\ushift{\_}{u}: \NCentailsup(S) \to \NCentailsup(R)$  is left adjoint to $\dshift{\_}{u}: \NCentailsup(R) \to \NCentailsup(S)$, i.e., for all $c\in \NCentailsup(R)$ and $b\in \NCentailsup(S)$ it holds
\[b \entails \dshift{c}{u} \iff  \ushift{b}{u} \entails c\]

\item  The join operator (concatenation of branches) ${\join} \of \NC{R}\times \NC{R} \to \NC{R}$ of \cref{def:join} is a categorical coproduct in $\NCentailsup(R)$, both product and coproduct in $\NCplain(R)$, product in $\NCrefl(R)$, and coproduct in $\NCpres(R)$. It is functorial in both arguments in all these categories.

\item The meet operator ${\meet} \of \NC{R}\times \NC{R} \to \NC{R}$  of \cref{def:meet} 
is a product in $\NCentailsup(R)$, but it is not functorial in the other categories, due to non-functoriality of downshift (point \ref{downshift-not-functorial} above).
\end{enumerate}
\end{propositionE}

\begin{proofE}~
\begin{enumerate}
\item Initiality of $\bot_R \in \NCentailsup(R)$ follows by \cref{prop:bottom}, because $False \entails c$ for any condition $c$. 
    Furthermore, given $c \in \NC(R)$, the only morphism $(id_R, \varnothing, \varnothing)\of \bot_R \to c$ is preservative, as $\varnothing \subseteq [1,0] \times [1,w]$ is a function, while the same  $(id_R, \varnothing, \varnothing): c \to \bot_r$ is reflexive, as $\varnothing \subseteq  [1,w] \times [1,0]$ is surjective and injective. 

\item By \cref{prop:inverse}, \inv{} acts as negation, and $c_1 \entails c_2$ iff $\neg c_2 \entails\neg c_1$.
\item By \cref{def:inverse}, $\inv c = (R, (id_R, c))$. If $m\of c_1 \to c_2$ is a morphism in $\NCplain(R)$, then $\inv{m} \isdef (\id_R, \setof{(1,1)}, N = \setof{N^{11}=\setof{m}})$ is a morphism from $\inv{c_2}$ to $\inv{c_1}$ in $\NCplain(R)$. Notice also that the 1,1 square of $\inv{m}$ is both reflective and preservative, as it is made of four identities, thus if $m$ is in $\NCrefl(R)$ (resp.\ $\NCpres(R)$) then so is $\inv{m}$. 
\item Obvious.

\item Functoriality of upshift on $\NCentailsup(R)$ follows by the same property of the existential construction of~\cite{bchk:conditional-reactive-systems}: we present a proof based on our notation. Let $u\of S \to R$ and \tagone $c_1 \entails c_2\in \NCentailsup(R)$. If $g \entails \ushift{c_1}{u}$, by point 2 of \cref{prop:shiftSat} \tagtwo  there is some $g'$ such that $g = u;g'$ and  $g' \entails c_1$. Thus by \tagone $g' \entails c_2$, and by \tagtwo $g \entails \ushift{c_2}{u}$, showing that $\ushift{\_}{u}$ is functorial on $\NCentailsup(R)$.

\smallskip
Given $m = (id_R,I,N) \of c_1 \to c_2$ in $\NCplain(R)$, let $\ushift{m}{u} \isdef  (id_S,I,N)\of \ushift{c_1}{u} \to \ushift{c_2}{u}$. This is well defined because the number of branches is preserved by upshift, as well as the subconditions. Every square of $\ushift{m}{u}$ induced by the sub-morphisms in $N$ is obtained by composing the corresponding square of $m$ with (on the top) the square having $u, id_S, id_R, u$: it can be shown that such composition preserves the defining property of reflection and preservation squares, as required.

\item Functoriality of downshift is known, see e.g.~\cite{bchk:conditional-reactive-systems}: we present a proof based on our notation. If \tagone $c_1 \entails c_2$, then $g \entails \dshift{c_1}{v} \stackrel{\tagtwo}{\iff} v; g \entails c_1 \stackrel{\tagone}{\Rightarrow} v;g \entails c_2 \stackrel{\tagtwo}{\iff} g \entails \dshift{c_2}{v}$, where \tagtwo holds by point 1 of \cref{prop:shiftSat}, showing that $\dshift{\_}{v}$ is functorial on $\NCentailsup(R)$.

\smallskip\cref{ex:downshift-not-functorial} shows an example of a morphism $m\of c_1\to c_2$ in $\NCplain$ that does not have an image under $\dshift{\_}{v}$; i.e., for which no morphism from $\dshift{c_1}v$ to $\dshift{c_2}v$ exists.

\item This adjunction is already stated in~\cite{bchk:conditional-reactive-systems}, given the semantic equivalence of the shift and existential operators defined there with the downshift and upshift operators. We present a proof based on our notation. Given $u\of R\to S$, $c \in \NC{R}$, and $b\in \NC{S}$,  we have to show that the following are equivalent:  
\smallskip
\begin{description}[nosep]
\item[\tagone] For all $m:S \to M$, $m \sat b \implies m \sat \dshift{c}{u}$;
\item[\tagtwo] For all $g:R\to G$, $g\sat \ushift{b}{u} \implies g \sat c$.
\end{description}

\begin{description}[nosep]
    \item[$\tagone \Rightarrow \tagtwo$] Assume \tagone and let $g:R\to G$ be such that $g\sat \ushift{b}{u}$. By point 2 of \cref{prop:shiftSat} there is some $g'$ such that $g = u;g'$ and $g' \sat b$. Thus by \tagone $g' \sat \dshift{c}{u}$, and by point 1 of \cref{prop:shiftSat} $g = u;g' \sat c$, as required.
\item[$\tagtwo \Rightarrow \tagone$] Assume \tagtwo and let $m:S \to M$ be such that $m \sat b$. By point 2 of \cref{prop:shiftSat} $u;m \sat \ushift{b}{u}$, thus by \tagtwo $u;m \sat c$. Therefore by point 1 of \cref{prop:shiftSat} $m \sat \dshift{c}{u}$, as required.
\end{description}

\item  $\join$ is coproduct in $\NCentailsup$ by \cref{prop:join}.

\smallskip
Instead the universal properties of $\join$ in $\NCplain$ and subcategories are mainly due to the properties of $\cat{Rel}$, the category of sets and relations, where products and coproducts coincide.
Let $c_i=\tupof{R,p^i_1\ccdots p^i_{w_i}} \in \AC R$ with $p^i_j = (a^i_j, b^i_j)$ for $i=1,2$, and $c_1\sqcup c_2= (R,p^1_1\ccdots p^1_{w_1}\, p^2_1\ccdots p^2_{w_2})$. Define morphisms $\pi_i\of c_1\join c_2 \to c_i$ as $\pi_1 = (id_R, \{(j,j)\mid j \in [1, w_1]\}, \setof{N^{jj}=\setof{id_{b^1_j}}})$ and $\pi_2 = (id_R, \{(j+w_1,j)\mid j \in [1, w_2]\}, \setof{N^{(j+w_1)j}=\setof{id_{b^2_j}}})$.  For each pair of morphisms $g = (id_R,I_g, N_g): c \to c_1$,  $f = (id_R,I_f, N_f): c \to c_2$ there is a unique morphism $(g,f) \of  c \to c_1 \join c_2$ such that $(g,f); \pi_1 = g$ and $(g,f); \pi_2 = f$, defined as $(g,f) = (id_R, I_g\cup \setof{(i, j+ w_1)\mid (i,j) \in I_f},  N_g \cup \setof{N^{i, j+w_1}\isdef N^{i,j} \mid N^{i,j} \in N_f})$. This makes $(c_1\join c_2, \pi_1, \pi_2)$ the product of $c_1$ and $c_2$ in $\NCplain(R)$. Note also that $\pi_1,\pi_2$ are reflective, as the relations are inverse functions and all induced squares are trivially reflective; furthermore, if $g\of c \to c_1$ and $f\of c \to c_2$ are reflective, so is $(g,f)$, thus the same is a product in $\NCrefl$.

\smallskip
Similarly,  injections $in_i\of c_i \to c_1 \join c_2$ are obtained from $\pi_1$ and $\pi_2$ by taking the opposite of the relation $I$ . It can be shown that $(c_1\join c_2, in_1, in_2)$ is the coproduct of $c_1$ and $c_2$ in $\NCplain(R)$, and since injections and induced morphisms are preservative, also in $\NCpres(R)$. 

\item $\meet$  is product in $\NCentailsup(R)$ by \cref{prop:meet}.
\end{enumerate}
\end{proofE}

\begin{example}\label{ex:downshift-not-functorial}
To see that downshift is not functorial (\cref{prop:functoriality}.\ref{downshift-not-functorial}), consider the conditions on the left hand of \cref{fig:downshift-not-functorial}. The dashed arrows between $c_1$ and $c_2$ together establish a valid reflective morphism $m\of c_1\to c_2$ with the identity as top arrow, establishing $c_2\entails c_1$. However, between the downshifted versions $\dshift{c_1}v$ and $\dshift{c_2}v$ (where the arrow $v$ is the one shown in the figure), no morphism exists: on the third level, there is no arrow from $\inline{\backedgedoubleedge z b x c a y}$ to $\inline{\backedgedoubleedge z a x c b y}$.\qed
\end{example}

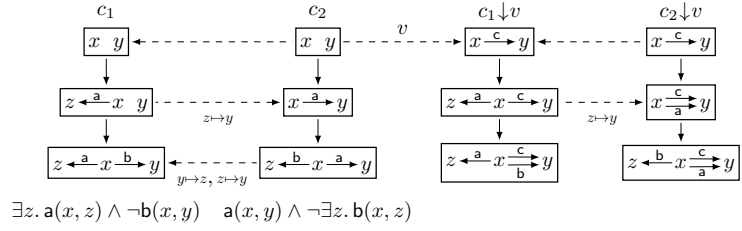
\begin{figure}\centering
\tikzscale{\begin{tikzpicture}[on grid]
  \node (c1) {$c_1$};
  \node[graph,below=.5 of c1] (a10) {\twonode x y};
  \node[graph,below=of a10] (a11) {\backedgenode z a x y};
  \node[graph,below=of a11] (a12) {\spangraph z a x b y};
  \node[below=.8 of a12] {$\exists z\st \la(x,z)\wedge \neg \lb(x,y)$};

  \path (a10) edge[->] (a11)
        (a11) edge[->] (a12);
  
  \node (c2) [right=3.5 of c1] {$c_2$};
  \node[graph,below=.5 of c2] (a20) {\twonode x y};
  \node[graph,below=of a20] (a21) {\oneedge x a y};
  \node[graph,below=of a21] (a22) {\spangraph z b x a y};
  \node[below=.8 of a22] {$\la(x,y) \wedge \neg\exists z\st \lb(x,z)$};

  \path (a20) edge[->] (a21)
        (a21) edge[->] (a22);

  \path (a10) edge[dashed,<-] (a20)
        (a21) edge[dashed,<-] node[below] {$\mapping{z&y}$} (a11)
        (a12) edge[dashed,<-] node[below] {$\mapping{y&z},\mapping{z&y}$} (a22);
        
  \node (c3)  [right=3 of c2] {$c_1\tdarr v$};
  \node[graph,below=.5 of c3] (a30) {\oneedge x c y};
  \node[graph,below=of a30] (a31) {\spangraph z a x c y};
  \node[graph,below=of a31] (a32) {\backedgedoubleedge z a x c b y};

  \path (a30) edge[->] (a31)
        (a31) edge[->] (a32);

  \path (a20) edge[dashed,->] node[above] {$v$} (a30);
  
  \node (c4) [right=3 of c3] {$c_2\tdarr v$};
  \node[graph,below=.5 of c4] (a40) {\oneedge x c y};
  \node[graph,below=of a40] (a41) {\doubleedge x c a y};
  \node[graph,below=of a41] (a42) {\backedgedoubleedge z b x c a y};

  \path (a40) edge[->] (a41)
        (a41) edge[->] (a42);

  \path (a30) edge[dashed,<-] (a40)
        (a41) edge[dashed,<-] node[below] {$\mapping{z&y}$} (a31);
\end{tikzpicture}}
\caption{Condition downshift is not functorial: there is a morphism from $c_1$ to $c_2$ but not from $\dshift{c_1}v$ to $\dshift{c_2} v$}.
\label{fig:downshift-not-functorial}
\end{figure}

\section{Conclusion, Related and Future work}
\label{sec:conclusion}

The main contributions of this paper can be summarized as follows:

\begin{itemize}[nosep]
\item 
After presenting (nested) conditions as in \cite{Rensink-FOL}, in \cref{sec:conditions} we introduce several operations on them showing the induced logical meaning on the corresponding FOL formulas. Such operators are shown to be complete.

\item In \cref{sec:morphisms} we propose an original notion of structural morphisms among conditions, which are shown to be composable and to have identities. We identify two subclasses of morphisms, the \emph{reflective} and the \emph{preservative} ones, which are shown to reflect and preserve satisfaction, respectively, and thus witness entailment between the formulas associated with the related conditions. We also show that, under mild assumptions, for each preservative morphism there is a reflective one in the opposite direction.

\item In \cref{sec:categories} we define several categories having conditions as objects and as arrows the various kinds of morphisms defined in \cref{sec:morphisms}. These are related with suitable  identity-on-object functors to the preorder of entailment among conditions. We also show some functoriality and universality properties of the constructions on conditions of \cref{sec:conditions} with respect to these categories. 
\end{itemize}
\paragraph{Discussion.}

The work reported here was primarily motivated by curiosity: given that entailment in the $\exists$-fragment of FOL (see \cref{sec:introduction}) can be fully characterised by structural morphisms between graphs, we try to push this idea by considering, not graphs, but nested conditions, which (over \cat{Graph}) are known to be expressively equivalent to full FOL. If successful, the resulting extended morphism would provide an understanding of \emph{why} one condition entails another; this can then potentially be manipulated and give rise to new forms of reasoning. One concrete case of this has serendipitously presented itself in the guise of \emph{Craig interpolants}, which we speculate (in the future work discussion below) might emerge from our structural morphisms. If one is interested in entailment checking per se, we do not think that searching for a morphism between two given conditions is likely to be superior to existing methods --- certainly not based on the precise form of morphism studied in this paper, which we have observed to be rather weak. (How weak precisely is again one of the questions we ask ourselves in the future work discussion.) 

One observation that has been made multiple times (see the PhD theses of Pennemann \cite{p:development-correct-gts} and Wallentin \cite{Wallentin-PhD}) is that the use of state-of-the-art satisfiability checkers such as \href{https://en.wikipedia.org/wiki/Z3_Theorem_Prover}{Z3} to the FOL translation of nested condition yields poorer performance than (academic) tools that work directly on the nested (graph) structure itself. To our knowledge, this phenomenon has not been studied in depth and so one should be careful in drawing conclusions; nevertheless, it suggests that the graph-based structure of nested conditions might actually have some performance advantages. The investigation in this paper can be seen as a contribution to the theoretical framework built on top of that graph-based structure.

\paragraph{Related work.}

We made precise along the paper the relationship between our morphisms and the preorder of entailment among conditions presented in~\cite{bchk:conditional-reactive-systems}. Unfortunately, space constraints do not allow us to discuss in detail the relationship with our previous work~\cite{DBLP:conf/birthday/Rensink025}. Anyway, we can explain at and abstract level why the morphisms proposed here are simpler and more general than those of~\cite{DBLP:conf/birthday/Rensink025}. In that paper we assumed that morphisms could relate only conditions over the same root, consistently with the standard notion of entailment (\cref{def:entailment}). This led to some complications, as in the inductive definition even if two conditions have the same root, two related subconditions (branch conditions) could have different roots, and these were reconciled with the use of so-called ``shifters''. In this paper, instead, we defined morphisms immediately between conditions over different roots (but related by an arrow) obtaining a simpler definition. We recovered the constraint on entailment by requiring the top arrow of a morphisms to be an identity, and we also generalized the standard notion of entailment for conditions having different roots (\cref{def:categories entailment}). 

As a matter of fact, in \cite{DBLP:conf/birthday/Rensink025} we also introduced a different kind of conditions, so-called \emph{span-based} ones, which were there posed as a variant that in principle allows more powerful reasoning. A natural question is whether the morphisms defined here for arrow-based conditions can be generalized to span-based ones. Likewise, the same question can be asked for other variants of nested conditions proposed in the literature, such as those of~\cite{Habel-FOL} and of~\cite{bchk:conditional-reactive-systems}.

\paragraph{Future work.}
The expressive power of the reflective morphisms we defined is quite limited, in the sense that they only explain a small fragment of entailment among formulas. A question that we intend to address in the future is whether there is any independent characterisation of the fragment of entailment that is explained by our condition morphisms and, connected to this, whether the proposed morphisms can be generalized to explain a larger fragment of entailment. 

Another topic we plan to explore, inspired by discussions with Barbara König, is related to the existence of \emph{Craig interpolants} for implications among formulas induced by conditions. Craig interpolants play an important role in improving abstraction-based reasoning (cf.\ \cite{Craig}), and being able to characterise them in the case of nested conditions could therefore strengthening the method studied in \cite{CEGAR-GT}. In particular, we are interested in the relationship between the existence of a morphism between two conditions and the existence of an interpolant for the corresponding implication. 

Finally, let us observe that in~\cite{bchk:conditional-reactive-systems,sksclo:coinductive-techniques-for-satisfiability} nested conditions are defined over an arbitrary category $\bC$, which allows to instantiate the framework beyond presheaf toposes: for example, to the category of graphs and injective morphisms, or to the category of left-linear cospans of an adhesive category. In some cases where pushouts do not (sufficiently) exist, their role can be taken over by (finite) \emph{sets of representative squares}. We intend to explore to what extent our assumptions on category $\bC$ can be relaxed in order to be able to apply our results to other, more general settings.

\paragraph{Acknowledgements.}
We would like to thank Filippo Bonchi for providing the original inspiration for this work, and together with Nicolas Behr and Barbara König for providing comments in earlier stages; and Nicolas and Barbara again for a series of inspiring brainstorm sessions on nested conditions.

\paragraph{Declaration on Generative AI.}
The authors have not employed any Generative AI tools.

\bibliography{arxiv-references}

@string{LNCS={LNCS}}

@book{MacLaneMoerdijk1992,
  author    = {Saunders Mac Lane and Ieke Moerdijk},
  title     = {Sheaves in Geometry and Logic: A First Introduction to Topos Theory},
  publisher = {Springer},
  year      = {1992},
  doi      = {10.1007/978-1-4612-0927-0}
}

@inproceedings{EhrigHKLRWC97,
  author    = {Hartmut Ehrig and
               Reiko Heckel and
               Martin Korff and
               Michael L{\"{o}}we and
               Leila Ribeiro and
               Annika Wagner and
               Andrea Corradini},
  title     = {Algebraic Approaches to Graph Transformation {II:} Single Pushout
               Approach and Comparison with Double Pushout Approach},
  booktitle = {Handbook of Graph Grammars and Computing by Graph Transformations,
               Vol. 1: Foundations},
  year      = {1997},
  pages = {247-312},
  crossref  = {handbook1},
  doi = {https://doi.org/10.1142/9789812384720_0004}
}

@proceedings{handbook1,
  editor    = {Grzegorz Rozenberg},
  title     = {Handbook of Graph Grammars and Computing by Graph Transformations,
               Volume 1: Foundations},
  year      = {1997},
  publisher = {World Scientific},
  isbn      = {9810228848},
  biburlNO    = {http://dblp.uni-trier.de/rec/bib/conf/gg/1997handbook},
  bibsource = {dblp computer science bibliography, http://dblp.org}
}

@article{Anjorin_Leblebici_Schürr_2016, 
  title ={20 Years of Triple Graph Grammars: A Roadmap for Future Research}, 
  volume ={73},
    DOI={10.14279/tuj.eceasst.73.1031}, 
  journal={Electronic Communications of the EASST},
  author={Anjorin, Anthony and Leblebici, Erhan and Schürr, Andy},
   year={2016}, 
   month={Apr.} 
   }

@inproceedings{DBLP:conf/birthday/Rensink025,
  author       = {Arend Rensink and
                  Andrea Corradini},
  editor       = {Nils Jansen and
                  Sebastian Junges and
                  Benjamin Lucien Kaminski and
                  Christoph Matheja and
                  Thomas Noll and
                  Tim Quatmann and
                  Mari{\"{e}}lle Stoelinga and
                  Matthias Volk},
  title        = {On Categories of Nested Conditions},
  extbooktitle    = {Principles of Verification: Cycling the Probabilistic Landscape -
                  Essays Dedicated to Joost-Pieter Katoen on the Occasion of His 60th
                  Birthday, Part {I}},
  booktitle    = {Principles of Verification: Cycling the Probabilistic Landscape, Part {I}},
  series       = LNCS,
  pages        = {393--418},
  publisher    = {Springer},
  year         = {2024},
  doi          = {10.1007/978-3-031-75783-9\_16},
  bibsource    = {dblp computer science bibliography, https://dblp.org}
}

@InProceedings{Rensink-FOL,
  author = 	 {Arend Rensink},
  editor       = {Hartmut Ehrig and
                  Gregor Engels and
                  Francesco Parisi{-}Presicce and
                  Grzegorz Rozenberg},
  title        = {Representing First-Order Logic Using Graphs},
  booktitle    = {Second International Conference on Graph Transformations (ICGT)},
  series       = LNCS,
  pages        = {319--335},
  publisher    = {Springer},
  year         = {2004},
  doi          = {10.1007/978-3-540-30203-2\_23},
}

@InProceedings{sksclo:coinductive-techniques-for-satisfiability,
  author = 	 {Lara Stoltenow and Barbara K{\"o}nig and
                  Sven Schneider and Andrea Corradini and Leen Lambers
                  and Fernando Orejas},
  editor       = {Rupak Majumdar and
                  Alexandra Silva},
  title        = {Coinductive Techniques for Checking Satisfiability of Generalized
                  Nested Conditions},
  booktitle    = {35th International Conference on Concurrency Theory (CONCUR)},
  series       = {LIPIcs},
  pages        = {1--39},
  extpublisher    = {Schloss Dagstuhl - Leibniz-Zentrum f{\"{u}}r Informatik},
  publisher    = {Leibniz-Zentrum f{\"{u}}r Informatik},
  year         = {2024},
  doi          = {10.4230/LIPICS.CONCUR.2024.39},
}

@PhdThesis{p:development-correct-gts,
  author = 	 {Karl-Heinz Pennemann},
  title = 	 {Development of Correct Graph Transformation Systems},
  school = 	 {Universit\"at Oldenburg},
  year = 	 2009,
  month =	 {May}
}

@Article{Habel-FOL,
  author = 	 {Habel, Annegret and Pennemann, Karl-Heinz},
  title        = {Correctness of high-level transformation systems relative to nested
                  conditions},
  journal      = {Math. Struct. Comput. Sci.},
  volume       = {19},
  number       = {2},
  pages        = {245--296},
  year         = {2009},
  doi          = {10.1017/S0960129508007202},
}

@article{NegativeAC,
	author = {Habel, Annegret and Heckel, Reiko and Taentzer, Gabriele},
	title = {Graph grammars with negative application conditions},
	journal = {Fundamenta Informaticae},
	issue_date = {December 1996},
	volume = {26},
	number = {3,4},
	month = dec,
	year = {1996},
	issn = {0169-2968},
	pages = {287--313},
  doi          = {10.3233/FI-1996-263404},
}

@Book{eept:fundamentals-agt,
  author =	 {Hartmut Ehrig and Karsten Ehrig and Ulrike Prange and Gabriele Taentzer},
  title = 	 {Fundamentals of Algebraic Graph Transformation},
  series       = {Monographs in Theoretical Computer Science},
  publisher    = {Springer},
  year         = {2006},
  doi          = {10.1007/3-540-31188-2},
}

@Article{ls:adhesive-journal,
  author =	 {Stephen Lack and Pawe{\l} Soboci\'{n}ski},
  title =	 {Adhesive and Quasiadhesive Categories},
  journal      = {{RAIRO} Theor. Informatics Appl.},
  volume       = {39},
  number       = {3},
  pages        = {511--545},
  year         = {2005},
  doi          = {10.1051/ITA:2005028},
}

@PhDThesis{Wallentin-PhD,
	Author="Lara Wallentin",
    Title="Verifying Graph Rewriting --- Coinductive Techniques for Nested Conditions",
    Year=2026,
    School = {University Duisburg-Essen, Germany},
}

@InProceedings{bchk:conditional-reactive-systems,
  author = 	 {{H.J.} Sander Bruggink and Rapha{\"e}l Cauderlier and
                  Mathias H{\"u}lsbusch and Barbara K{\"o}nig},
  title = 	 {Conditional Reactive Systems},
  booktitle =    {{IARCS} Annual Conference on Foundations of Software Technology and
                    Theoretical Computer Science (FSTTCS)},
  editor       = {Supratik Chakraborty and
                    Amit Kumar},
  series       = {LIPIcs},
  pages        = {191--203},
  extpublisher    = {Schloss Dagstuhl - Leibniz-Zentrum f{\"{u}}r Informatik},
  publisher    = {Leibniz-Zentrum f{\"{u}}r Informatik},
  year         = {2011},
  doi          = {10.4230/LIPICS.FSTTCS.2011.191},
  volume = 	 {13},
}

@phdthesis{p:correct-GTSs-thesis,
  author       = {Karl{-}Heinz Pennemann},
  title        = {Development of correct graph transformation systems},
  school       = {University of Oldenburg, Germany},
  year         = {2009},
  url          = {https://nbn-resolving.org/urn:nbn:de:gbv:715-oops-9483},
  urn          = {urn:nbn:de:gbv:715-oops-9483},
  bibsource    = {dblp computer science bibliography, https://dblp.org}
}

@inproceedings{DBLP:conf/stoc/ChandraM77,
  author       = {Ashok K. Chandra and
                  Philip M. Merlin},
  editor       = {John E. Hopcroft and
                  Emily P. Friedman and
                  Michael A. Harrison},
  title        = {Optimal Implementation of Conjunctive Queries in Relational Data Bases},
  booktitle    = {9th Annual {ACM} Symposium on Theory of Computing (STOC)},
  pages        = {77--90},
  publisher    = {{ACM}},
  year         = {1977},
  doi          = {10.1145/800105.803397},
  bibsource    = {dblp computer science bibliography, https://dblp.org}
}

@article{AzziCR19,
  author       = {Guilherme Grochau Azzi and
                  Andrea Corradini and
                  Leila Ribeiro},
  title        = {On the essence and initiality of conflicts in {M}-adhesive transformation
                  systems},
  journal      = {J. Log. Algebraic Methods Program.},
  volume       = {109},
  year         = {2019},
  doi          = {10.1016/J.JLAMP.2019.100482},
  bibsource    = {dblp computer science bibliography, https://dblp.org}
}

@book{DBLP:books/sp/HeckelT20,
  author       = {Reiko Heckel and
                  Gabriele Taentzer},
  title        = {Graph Transformation for Software Engineers - With Applications to
                  Model-Based Development and Domain-Specific Language Engineering},
  publisher    = {Springer},
  year         = {2020},
  doi          = {10.1007/978-3-030-43916-3},
  isbn         = {978-3-030-43915-6},
  bibsource    = {dblp computer science bibliography, https://dblp.org}
}

@misc{mathexchange,
  author="Daniel Schepler",
  Title="Does the existence of a commuting diagonal in a pushout imply that one leg is epi?",
  Howpublished="Answer on \href{https://math.stackexchange.com/questions/5129159/does-the-existence-of-a-commuting-diagonal-in-a-pushout-imply-that-one-leg-is-ep}{math.stackexchange.com}",
  year = {2026},
}

@inproceedings{Craig,
  author    = {Thomas A. Henzinger and
               Ranjit Jhala and
               Rupak Majumdar and
               Kenneth L. McMillan},
  title     = {Abstractions from proofs},
  booktitle = {31st Symposium on Principles
               of Programming Languages (POPL)},
  pages     = {232--244},
  year      = {2004},
  doi       = {10.1145/964001.964021},
  bibsource = {dblp computer science bibliography, https://dblp.org}
}

@inproceedings{CEGAR-GT,
  author       = {Barbara K{\"{o}}nig and
                  Arend Rensink and
                  Lara Stoltenow and
                  Fabian Urrigshardt},
  editor       = {J{\"{o}}rg Endrullis and
                  Matthias Tichy},
  title        = {Counterexample-Guided Abstraction Refinement for Generalized Graph
                  Transformation Systems},
  booktitle    = {18th International Conference on Graph Transformation (ICGT)},
  series       = LNCS,
  pages        = {157--176},
  publisher    = {Springer},
  year         = {2025},
  doi          = {10.1007/978-3-031-94706-3\_8},
  bibsource    = {dblp computer science bibliography, https://dblp.org}
}

@InProceedings{GCM-original,
  author={Arend Rensink and Andrea Corradini},
  title={Structural Morphism for Nested Conditions},
  booktitle={Graph Computational Models (GCM)},
  series = {CEUR Workshop Proceedings},
  issn = {1613-0073},
  note="In press",
  year={2026},
}

\end{document}